\documentclass{article}
\usepackage{microtype}
\usepackage{graphicx}
\usepackage{subfig}
\usepackage{graphicx}
\usepackage{float}
\usepackage{multirow}
\usepackage{bm}
\usepackage[figuresright]{rotating}
\usepackage{booktabs} 

\usepackage{hyperref}
\usepackage{textcomp}
\usepackage{stfloats}
\usepackage{url}
\usepackage{bbding}
\usepackage{color}
\usepackage{verbatim}
\usepackage{hyperref}
\hypersetup{
hidelinks,
colorlinks=true,
linkcolor=red,
citecolor=cyan,
urlcolor = magenta
}
\usepackage{orcidlink}
\usepackage{graphicx}
\usepackage{cite}
\usepackage{microtype}
\usepackage{graphicx}
\usepackage{subfig}
\usepackage{graphicx}
\usepackage{float}
\usepackage{multirow}
\usepackage{bm}
\usepackage{booktabs}
\usepackage{amsmath}
\usepackage{amssymb}
\usepackage{mathtools}
\usepackage{amsthm}

\usepackage[accepted]{icml2023}

\usepackage{amsmath}
\usepackage{amssymb}
\usepackage{mathtools}
\usepackage{amsthm}

\usepackage[capitalize,noabbrev]{cleveref}

\theoremstyle{plain}
\newtheorem{theorem}{Theorem}[section]

\theoremstyle{definition}

\theoremstyle{remark}

\usepackage[textsize=tiny]{todonotes}

\icmltitlerunning{MLIC++: Linear Complexity Multi-Reference Entropy Modeling for Learned Image Compression}

\begin{document}

\twocolumn[
\icmltitle{MLIC$^{++}$: Linear Complexity Multi-Reference Entropy Modeling for Learned Image Compression}




\begin{icmlauthorlist}
\icmlauthor{Wei Jiang}{pku}
\icmlauthor{Jiayu Yang}{pcl}
\icmlauthor{Yongqi Zhai}{pku,pcl}
\icmlauthor{Feng Gao}{pkuart}
\icmlauthor{Ronggang Wang}{pku,pcl}
\\\tt{wei.jiang1999@outlook.com, rgwang@pkusz.edu.cn}
\end{icmlauthorlist}

\icmlaffiliation{pku}{Guangdong Provincial Key Laboratory of Ultra High Definition Immersive Media Technology, Peking
University Shenzhen Graduate School}
\icmlaffiliation{pkuart}{School of Arts, Peking University}
\icmlaffiliation{pcl}{Pengcheng Laboratory}

\icmlcorrespondingauthor{Ronggang Wang}{rgwang@pkusz.edu.cn}

\icmlkeywords{Machine Learning, ICML}

\vskip 0.3in
]



\printAffiliationsAndNotice{} 

\begin{abstract}
  The latent representation in learned image compression encompasses channel-wise,
    local spatial, and global spatial correlations, which are essential 
    for the entropy model to capture for conditional entropy minimization.
    Efficiently capturing these contexts within a single entropy model,
    especially in high-resolution image coding, presents a challenge
    due to the computational complexity of existing global context modules.
    {To address this challenge, we propose the Linear Complexity Multi-Reference Entropy Model (MEM$^{++}$). 
    Specifically, the latent representation is partitioned into multiple slices. }
    {For channel-wise contexts, previously compressed slices serve as the context for compressing a particular slice. 
    For local contexts, we introduce a shifted-window-based checkerboard attention module. 
    This module ensures linear complexity without sacrificing performance. 
    For global contexts, we propose a linear complexity attention mechanism. 
    It captures global correlations by decomposing the softmax operation, 
    enabling the implicit computation of attention maps from previously decoded slices.}
    {Using MEM$^{++}$ as the entropy model, we develop the image compression method MLIC$^{++}$}. Extensive experimental results demonstrate that MLIC$^{++}$ achieves state-of-the-art performance, 
    reducing BD-rate by $13.39\%$ on the Kodak dataset compared to VTM-17.0 in 
    Peak Signal-to-Noise Ratio (PSNR). Furthermore, MLIC$^{++}$ exhibits 
    {linear computational complexity and memory consumption with resolution}, 
    making it highly suitable for high-resolution image coding.
    Code and pre-trained models are available at \texttt{\url{https://github.com/JiangWeibeta/MLIC}}.
    Training dataset is available at \texttt{\url{https://huggingface.co/datasets/Whiteboat/MLIC-Train-100K}}.
  \end{abstract}
  
  \section{Introduction}
Due to the rise of social media, tens of millions of images
are generated and transmitted on the web every second.
\begin{figure*}
  \centering
  \includegraphics[width=0.43\linewidth]
  {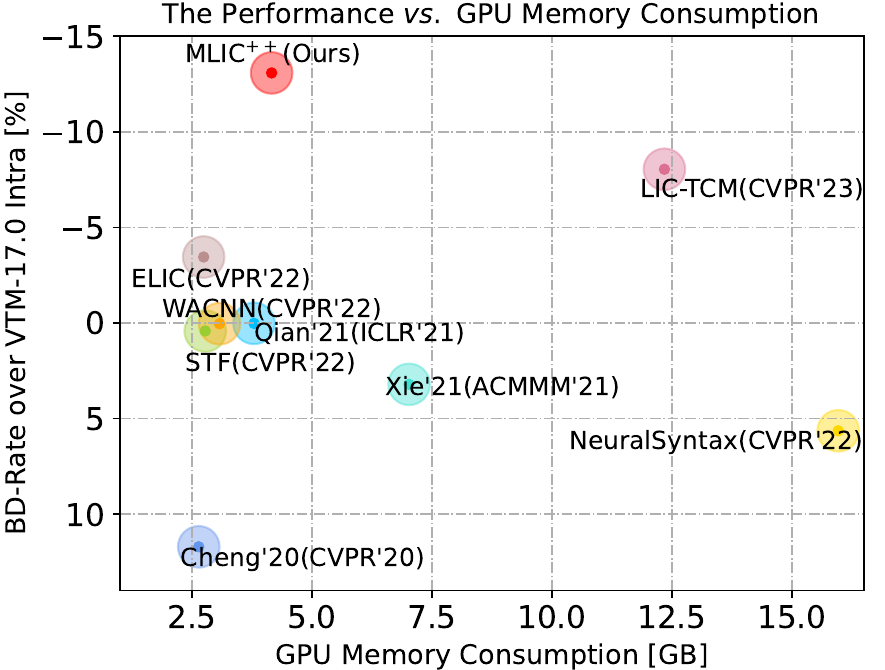}
  \includegraphics[width=0.40\linewidth]
  {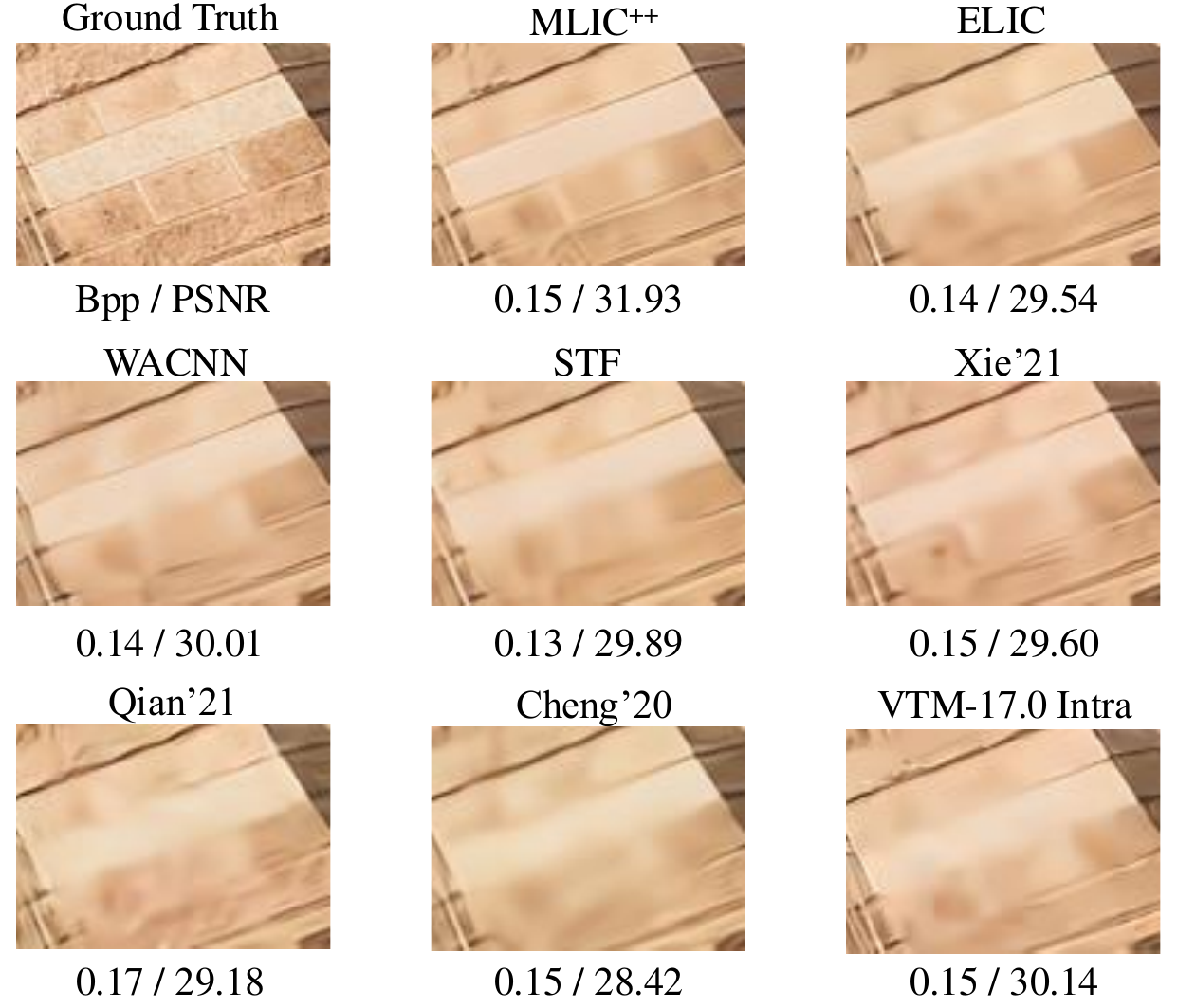}
  \caption{Left: BD-Rate-GPU Memory Consumption during inference on CLIC Professional Valid~\cite{CLIC2020} with 2K resolution.
  Our MLIC$^{++}$ achieves a better trade-off between performance and GPU memory consumption.
  {Right: Reconstruction comparison on ``vita-vilcina-3055'' from CLIC Professional Valid~\cite{CLIC2020} dataset. 
  The reconstruction of MLIC$^{++}$ has the best visual quality.}
  }
  \label{fig:ctx_compare_new}
\end{figure*}
In order to conserve bandwidth, service providers are compelled to seek more efficient and effective image compression methods.
Although traditional coding methods like JPEG~\cite{pennebaker1992jpeg},
JPEG2000~\cite{maryline1999jpeg2000}, AVC~\cite{wiegand2003avc}, HEVC~\cite{sullivan2012overview}, and VVC~\cite{bross2021vvc}
have achieved commendable performance, their design relies on manual design for each module.
This lack of joint optimization hampers their ability to fully exploit the potential for further advancements in image compression.\par
Recently, various learned image compression models
~\cite{wu2021learned,guo2021causal,balle2016end,theis2017lossy,balle2018variational,minnen2018joint,hu2020coarse,ma2020iwave}
have emerged, showcasing impressive performance gains.
Notably, certain learned image compression models~\cite{minnen2020channel,cheng2020learned,zou2022the,qian2020learning,xie2021enhanced,he2022elic,jiang2022mlic,chen2021nic,gao2021neural,chen2022two,koyuncu2022contextformer,duan2023lossy,liu2023learned,jiang2023slic,duan2023qarv} are already comparable to
the advanced traditional method VVC.
These models predominantly rely on auto-encoders or variational auto-encoders~\cite{kingma2014vae},
and follow a process that involves transform, quantization, entropy coding, and inverse transform.
Entropy coding plays an important role
in boosting model performance. An entropy model is utilized to estimate
the entropy of the latent representation.
A powerful and accurate entropy model usually leads to fewer bits.
Expanding contexts of entropy model in learned codecs plays the same role as expanding prediction modes in traditional codecs.\par
State-of-the-art learned image compression models~\cite{cheng2020learned,minnen2020channel,xie2021enhanced,zhu2022transformerbased,he2022elic,jiang2022mlic,jiang2023slic,lin2023multistage}
commonly enhance the entropy model by incorporating a hyper-prior module~\cite{balle2018variational} or a context module~\cite{minnen2018joint}.
These additional modules enable the estimation of conditional entropy and the utilization of conditional probabilities for entropy coding.
Context modules usually model probabilities and correlations in different dimensions,
including local spatial context module,
global spatial context module, and channel-wise context module.
However, the current global context modules rely on computationally intensive \textit{quadratic} complexity computations,
which consume \textit{huge} GPU memories and have slower encoding and decoding speed,
imposing limitations on the potential for high-resolution image coding.
Furthermore,
effectively capturing local, global, and channel-wise contexts with \textit{acceptable} even \textit{linear} complexity
within a single entropy model remains a challenge.
To overcome the aforementioned limitations,
we propose a novel \textit{linear complexity}  multi-reference entropy model.
This entropy model effectively captures local spatial, global spatial, and channel-wise contexts with \textit{linear} complexity
and can be employed for efficient high-resolution image coding, which is denoted as MEM$^{++}$
to differ from our prior work~\cite{jiang2022mlic} presented at ACMMM 2023.
Based on MEM$^{++}$, we introduce MLIC$^{++}$,
which achieve state-of-the-art performance.\par
In our approach, the latent representations is divided into multiple slices along the channel dimension.
When compressing a particular slice, the previously compressed slices serve as its channel-wise contexts,
which are extracted by a channel-wise context module.
Local and global context modeling are conducted separately for each slice.
The utilization of an auto-regressive local context module~\cite{minnen2018joint,van2016conditional} leads to serial decoding,
while a checkerboard context module~\cite{he2021checkerboard} facilitates
two-pass parallel decoding by dividing the latent representations into anchor and non-anchor parts.
However, it is worth noting that the checkerboard context module may result in a performance degradation of up to $4\%$~\cite{qian2022entroformer}.
To address this issue, we propose a novel overlapped checkerboard window attention with
\textit{linear} complexity, which
further enhances the local context capturing while retaining two-pass decoding.
Some previous methods focus on global context modeling~\cite{qian2020learning,guo2021causal},
which typically involve \textit{quadratic} complexity or the utilization of additional bits to store global similarity as side information.
Additionally, these global context modules often collaborate with serial local spatial context modules, further increasing the computational complexity.
Assuming comparable spatial correlations across different slices,
we initially calculate the attention map of the previously decoded $i-1$-th slice in a \textit{vanilla} approach.
This attention map is utilized to predict the global correlations within the $i$-th slice.
However, in the \textit{vanilla} attention mechanism, the softmax operation dictates the order of computation among tensors,
where the attention map, product of queries and keys are required to be computed first.
To circumvent the \textit{quadratic} complexity,
we employ the decomposition of a softmax operation
into two independent softmax operations such that the product of keys and values can be computed first,
resulting in \textit{linear} complexity.
The proposed linear complexity attention-based global context modules capture 
global correlations in an \textit{implicit} way, as there is no need to directly compute
the attention map during training and testing.
In addition, we also propose the \textit{linear} complexity inter-slice global spatial
context modules to explore the global correlations in all preceding slices.
The \textit{linear complexity} context allows our model to have a \textit{linear} relationship
between consumed GPU memory and resolution \textit{without} additional bits, while
having the performance gain that comes from the global contexts.
Ultimately, we integrate the channel, local, intra-slice, inter-slice global contexts,
along with the side information for multi-reference entropy modeling.
Our contributions are summarized as follows:
\begin{itemize}
\item To address the degradation associated with checkerboard context modeling
while preserving the benefits of two-pass decoding, we devise a novel
approach called shifted window-based checkerboard attention with \textit{linear} complexity.
This technique enables us to capture local spatial contexts more effectively.
{\item To address the degradation associated with checkerboard context modeling
while preserving the benefits of two-pass decoding, we devise a novel
approach called shifted window-based checkerboard attention with \textit{linear} complexity.
This technique enables us to capture local spatial contexts more effectively.
\item For enhanced context modeling efficiency, 
we decompose the softmax operations in vanilla attention into two independent
softmax operations. This reduces computational complexity to linear without compromising performance. 
Additionally, we explore the global correlations
\item 
We design \textit{linear complexity}  We design \textit{linear complexity}  multi-reference entropy model MEM$^{++}$ which
captures local spatial, global spatial and channel contexts, as well as
hyper-prior side information. 
Using MEM$^{++}$ as the entropy model,
we develop MLIC$^{++}$, which achieves state-of-the-art performance with linear complexity.
Our proposed MLIC$^{++}$ achieved a better trade-off between complexity and performance as depicted in Fig.~\ref{fig:ctx_compare_new}.}
\end{itemize}\par
In comparison to our previous work presented at ACMMM 2023~\cite{jiang2022mlic}, MLIC$^{++}$ introduces several significant advancements.
The primary distinction lies in the utilization of the proposed \textit{linear}
complexity global spatial context modules without sacrificing performance,
as opposed to the \textit{quadratic} complexity observed in our prior work.
This achievement is primarily attributed to the division
of the softmax operation, which eliminates the need for
a specific order of tensor computation. Our proposed modules
incorporate advanced techniques such as learnable position embedding and
depth-wise residual bottlenecks~\cite{jiang2023slic}.
Furthermore, MLIC$^{++}$ captures inter-slice global correlations from all previous slices,
in contrast to our previous work~\cite{jiang2022mlic},
which only considers correlations within the previous one slice.
MLIC$^{++}$
exhibits several advantages,
including reduced GPU memory consumption and faster encoding and decoding speed.
The \textit{linear} complexity makes our MLIC$^{++}$ highly suitable for high-resolution image coding.
{
  To demonstrate the superiority of MLIC$^{++}$, we conduct comprehensive performance and complexity comparisons against existing methods~\cite{cheng2020learned,minnen2020channel,qian2020learning,xie2021enhanced,qian2022entroformer,zhu2022transformerbased,wang2022neural,zhu2022unified,zou2022the,he2022elic,liu2023learned} across multiple datasets~\cite{CLIC2020,kodak,tecnick2014TESTIMAGES} and resolutions~\cite{liu2020comprehensive}, 
  extending beyond the preliminary experiments presented in the workshop version~\cite{jiang2023mlicpp}.}
These advancements in MLIC$^{++}$ contribute to the field of image compression
by offering improved efficiency and performance,
while maintaining high-quality compression capabilities.

\section{Related Works}
\label{sec:related}
\subsection{Learned Image Compression}
Learned image compression~\cite{balle2016end,theis2017lossy} aims to optimize the trade-off
between distortion $\mathcal{D}$ and entropy, where entropy is typically measured in terms of {bit-rate~\cite{theis2017lossy,yang2023introduction,balle2016end,balle2018variational,balle2020nonlinear}} $\mathcal{R}$.
Large bit-rate usually leads to lower distortion .
Lagrange multiplier $\lambda$ is employed to adjust the weight of
distortion to control the target bit-rate. The optimization target is
\begin{equation}\label{eq:rd}
    \mathcal{L} = \mathcal{R} + \lambda \mathcal{D}.
\end{equation}\par
The fundamental learned image compression framework~\cite{balle2016end,theis2017lossy}
is based on the an auto-encoder with a rate penalty. This framework comprises an
analysis transform $g_a$, a quantization function $Q$, a synthesis transform $g_s$ and
an entropy model to estimate rates. The process can be formulated as:
\begin{equation}
    \boldsymbol{y} = g_a(\boldsymbol{x};\theta), \hat {\boldsymbol{y}} = Q(\boldsymbol{y}), \hat {\boldsymbol{x}} = g_s(\hat {\boldsymbol{y}};\phi),
\end{equation}
where $\boldsymbol x$ represents the input image, $g_a$ transform the $\boldsymbol x$ to
compact latent representation $\boldsymbol y$. $\boldsymbol y$
is quantized to $\hat {\boldsymbol{y}}$ for entropy coding.
$\hat {\boldsymbol{x}}$ represents the decompressed image.
$\theta$ and $\phi$ are parameters of $g_a$ and $g_s$.
Since quantization is non-differentiable, it can be addressed during training by either adding uniform noise
$\mathcal{U}(-0.5, 0.5)$~\cite{balle2016end,balle2018variational} or using the straight-through estimator (STE)~\cite{theis2017lossy}.
In particular, when uniform noise is added, the rate-distortion optimization
target in Equation~\ref{eq:rd} is equivalent to evidence lower bound (ELBO) optimization in variational
auto-encoders~\cite{kingma2014vae}.
To enhance non-linearity, Generalized Divisive Normalization (GDN)~\cite{balle2015gdn} layers or its variants~\cite{qian2020learning} are employed.
Additionally, self-attention~\cite{vas2017attention,liu2021swin,lu2021tic, zou2022the,liu2023learned,guo2021causal},
ensemble techniques~\cite{wang2020ensemble}, and block partition~\cite{wu2021learned} are utilized in transform modules for
more compact latent representations.
In the basic model, a factorized or
a non-adaptive density entropy model is adopted.\par
In subsequent works, a hyper-prior module~\cite{balle2018variational} is introduced to extract side information $\hat {\boldsymbol{z}}$ from $\boldsymbol y$.
The hyper-prior model estimates the distribution of $\hat {\boldsymbol{y}}$ from
$\hat {\boldsymbol{z}}$. A univariate Gaussian distribution is commonly employed for the hyper-prior.
Some works extend it to a mean-scale Gaussian distribution~\cite{minnen2018joint},
asymmetric Gaussian distribution~\cite{cui2021asym},
Gaussian mixture model~\cite{cheng2020learned,liu2020unified}, and Gaussian-Laplacian-Logistic mixture model~\cite{fu2023learned}
for more flexible distribution modeling.
\subsection{Context-based Entropy Modeling}
\label{sec:related:context}
Numerous approaches~\cite{minnen2018joint,minnen2020channel,qian2020learning} have been proposed to improve the accuracy of context modeling in learned image compression.
These methods encompass various types of context modules, including local spatial, global spatial, and channel-wise context modules.\par
Local spatial context modules aim to capture correlations between adjacent symbols.
For instance, Minnen~\textit{et al.}~\cite{minnen2018joint} utilize a pixel-cnn-like~\cite{van2016conditional}
masked convolutional layer to capture local correlations between
$\hat {\boldsymbol{y}}_i$ and symbols $\hat {\boldsymbol{y}}_{<i}$,
resulting in serial decoding.
He \textit{et al.}~\cite{he2021checkerboard}
divide latent representation $\hat {\boldsymbol{y}}$ into
anchor part $\hat {\boldsymbol{y}}_a$ and non-anchor part $\hat {\boldsymbol{y}}_{na}$,
employing a checkerboard convolution to extract contexts of
$\hat {\boldsymbol{y}}_{na}$ from $\hat {\boldsymbol{y}}_a$, thereby
achieving two-pass parallel decoding.
\par
On the other hand, some approaches focus on modeling correlations between distant symbols.
In~\cite{qian2020learning}, neighboring left and top symbols serve as bases for computing the similarity
between the target symbol and its previous symbols.
Guo~\textit{et al.}~\cite{guo2021causal} employ
the $L2$ distances of symbols to predict global casual dependencies among symbols.
In~\cite{kim2022joint}, the side information
is divided into global side information and local side information, introducing additional bits.
However, these global context modules are typically combined with serial auto-regressive context modules,
which further increase decoding latency.
Moreover, existing global context modules~\cite{qian2020learning,guo2021causal,jiang2022mlic} often exhibit \textit{quadratic} complexity,
making them challenging to apply in high-resolution image coding.
Alternatively, they rely on extra side information~\cite{kim2022joint}, which increases the bit-rate.
\par
Minnen \textit{et al.}~\cite{minnen2020channel} model contexts between channels.
$\hat {\boldsymbol{y}}$ is evenly divided to slices. The current slice
$\hat {\boldsymbol{y}}^i$ is conditioned on previously decoded slices
$\hat {\boldsymbol{y}}^{<i}$. To address the uneven distribution of information among slices,
an unevenly grouped channel-wise context module is introduced in~\cite{he2022elic}.\par
While some local and channel-wise context modules~\cite{ma2021cross,he2022elic,jiang2022mlic}
have demonstrated impressive performance,
effectively capturing local, global, and channel-wise contexts with \textit{acceptable} even \textit{linear} complexity
within a single entropy model remains a challenge.
Addressing these correlations has the potential to further enhance the performance of image compression models.
\begin{figure*}[t]
\centering
\includegraphics[width=0.47\linewidth]
{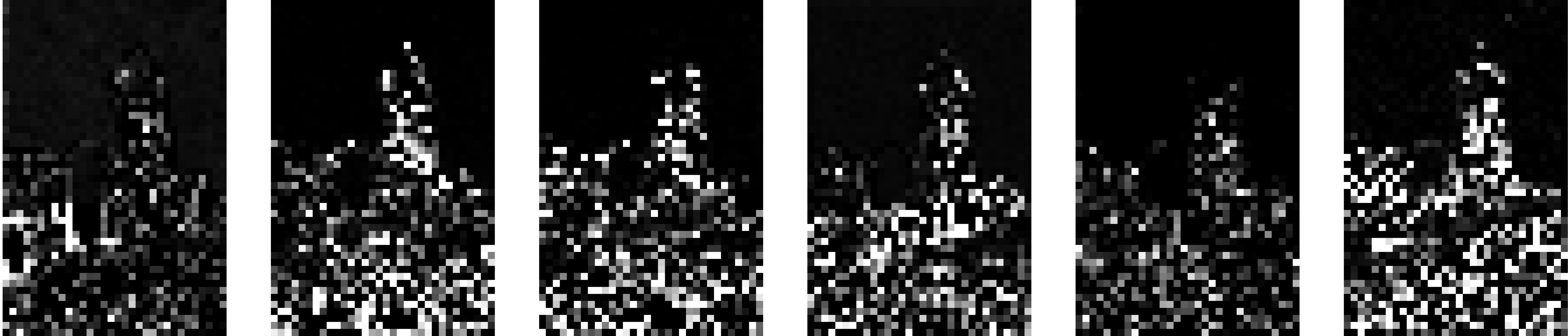}
\includegraphics[width=0.47\linewidth]
{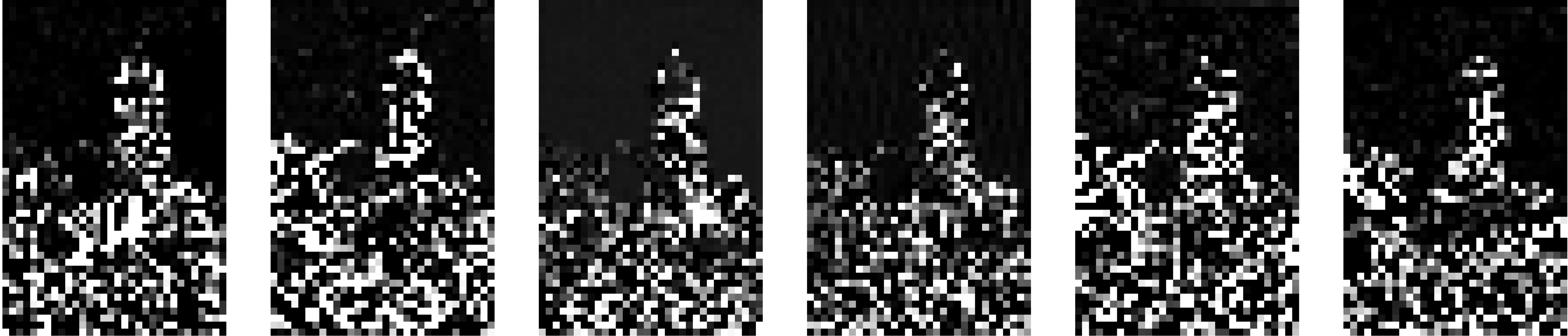}
\caption{Visualization of channels of latent representation of Kodim19
extracted by Cheng'20~\cite{cheng2020learned} (optimized for MSE, $\lambda=0.0483$) to illustrate channel-wise redundancy.
{These channels are nearest-neighbor upsampled for visualization}.}
\label{fig:cosine}
\end{figure*}
\begin{figure}[t]
  \centering
  \includegraphics[width=0.62\linewidth]
  {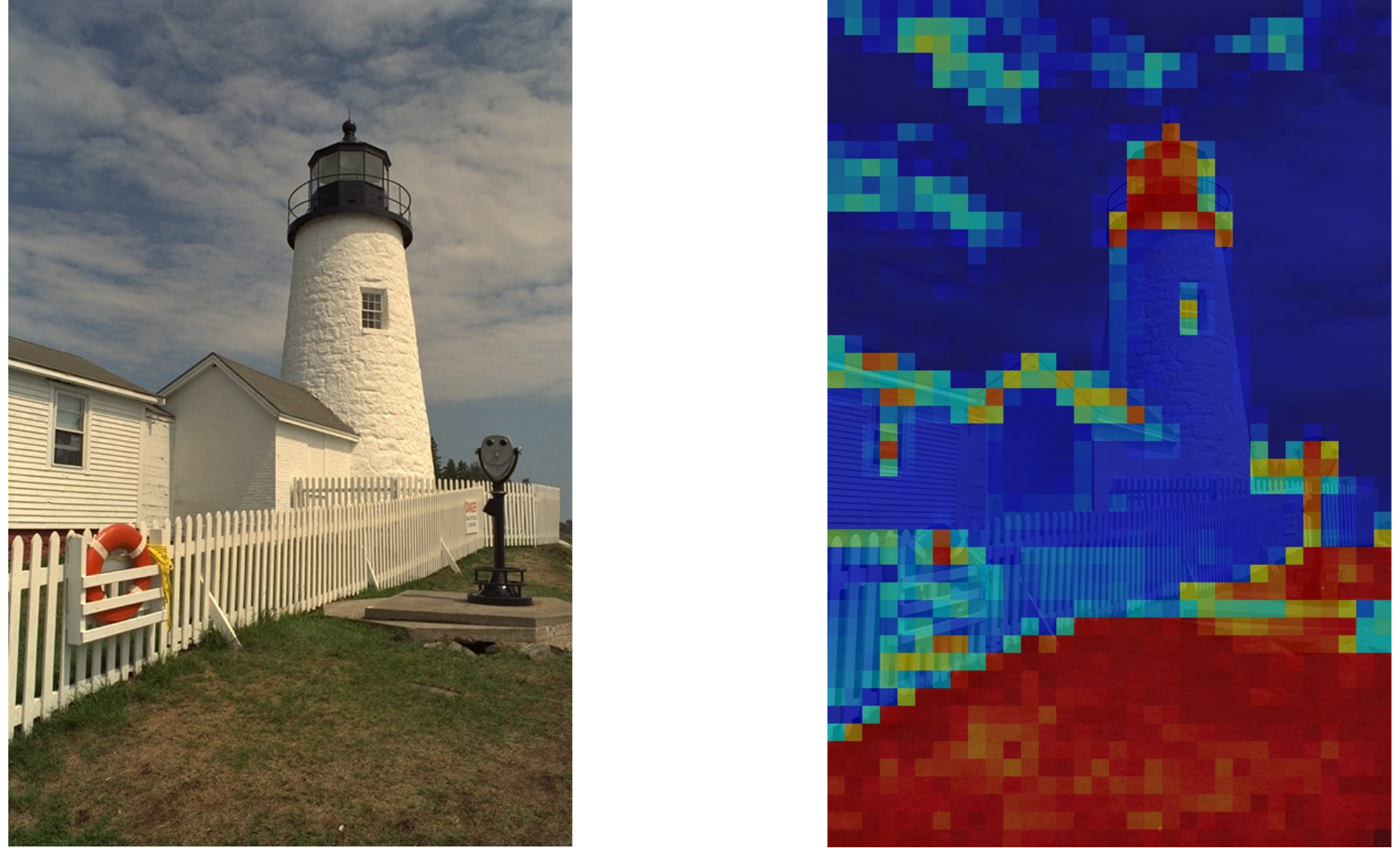}
  \caption{Heatmap of spatial cosine similarity of latent representation
  of Kodim19 extracted by Cheng'20~\cite{cheng2020learned} (optimized for MSE, $\lambda=0.0483$) to visualize global spatial and local spatial redundancy.
  {The heatmap is nearest-neighbor upsampled for visualization}.}
  \label{heatmap}
  \end{figure}
\begin{table*}[t]
\small
\centering
\begin{tabular}{c|c|c|c}
\toprule
Notations                                         & Explanation  & Notations                                         & Explanation  \\ \midrule
$\boldsymbol{x},\hat{\boldsymbol{x}}$                                  & Input and decoded image    & $\boldsymbol{y},\hat{\boldsymbol{y}}$        & Non-quantized and quantized latent representation                  \\\midrule
$\hat{\boldsymbol{y}}^i $                         & The $i$-th slice of $\hat{\boldsymbol{y}}$ & $\boldsymbol{z},\hat{\boldsymbol{z}}$  &     Non-quantized and quantized side information              \\\midrule
$\hat{\boldsymbol{y}}_{ac},\hat{\boldsymbol{y}}_{na}$                       & Anchor and non-anchor part of $\hat{\boldsymbol{y}}$   & $g_{ep}$     &  Entropy parameter module  \\\midrule
$\mu,\sigma$                                      & Mean and scale of $\hat{\boldsymbol{y}}$    & $g_a, g_s$     &  Analysis and synthesis transform       \\\midrule
$h_a,h_s$                                         & Hyper analysis and synthesis &  $g_{ch}$     & Channel-wise context module             \\\midrule
$g_{lc,ckbd}$                                     & Vanilla checkerboard context module       &   $g_{lc,attn}$                                     & Shifted Window-based Checkerboard Attention           \\\midrule
$g_{gc,intra}$                                    & Intra-slice global spatial context module    &   $g_{gc,inter}$                                    & Inter-slice global spatial context module     \\\midrule
$ {\boldsymbol{\Phi}}_{h},\boldsymbol{\Phi}_{ch},{\boldsymbol{\Phi}}_{lc}$               & Hyper-prior, channel-wise and local spatial context   &  $ {\boldsymbol{\Phi}}_{gc,intra}$                 & Intra-slice global spatial context   \\\midrule
$ {\boldsymbol{\Phi}}_{gc,inter}$                 & Inter-slice global spatial context & MEM $^{++}$  & Linear complexity multi-reference entropy model  \\\midrule 
$M, N, S$                                         & Channel number of $\boldsymbol{y}$, $\boldsymbol{z}$, and $\hat{\boldsymbol{y}}^i$      &     $K$                                               & Kernel size of local spatial context module  \\\midrule

\end{tabular}
\label{tab:notation}
\caption{Explanations of notations.}
\end{table*}
\begin{table}[t]
\footnotesize
\centering
\setlength{\tabcolsep}{1mm}{
\begin{tabular}{lccccccccc}
  \toprule
   $N$    & $M$    & $S$     & $K$    & Entropy Model       \\ \midrule
     $192$  & $320$  &$32$  & $5$   &   MEM$^{++}$($g_{lc,attn}, g_{ch}, g_{gc,intra}, g_{gc,inter}$)   \\\midrule
\end{tabular}}
\caption{Settings of MLIC$^{++}$ and MEM$^{++}$.}
\label{tab:settings}
\end{table}
\section{Method}
\label{sec:method}
\subsection{Motivation}
According to information theory, the conditional entropy is bounded by the entropy:
\begin{equation}
    H(\hat{\boldsymbol{y}}) \geq H(\hat{\boldsymbol{y}}|\boldsymbol{ctx}),
\end{equation}
where $H$ denotes Shannon entropy, $\boldsymbol{ctx}$ is the context of $\hat {\boldsymbol{y}}$.
Exploiting correlations in $\hat {\boldsymbol{y}}$ results in bit savings.\par
In Fig.~\ref{fig:cosine} and Fig.~\ref{heatmap},
channel-wise correlations and spatial correlations in
latent representation of Kodim19 extracted by Cheng'20~\cite{cheng2020learned} are illustrated.\par
Fig.~\ref{fig:cosine} visualizes the features of several channels, revealing their significant similarity.
However, capturing such correlations poses a challenge for spatial context modules,
as they employ the same \textit{mask} for all channels during context extraction.
Consequently, certain correlations may not be fully captured.\par
In Fig.~\ref{heatmap}, cosine similarity between
each symbol and the symbol in the bottom right corner are visualized.
Symbols with the same color exhibit a high degree of correlation.
Neighbouring symbols have a very high degree of similarity.
This observation emphasizes the necessity of a local context module.
Furthermore, a global context module is required to capture the correlations
between symbols in the bottom-left corner and those in the bottom-right corner,
where the grass features share similarities. Additionally, 
the complexity of global context capturing should be carefully considered and minimized for
high-resolution image coding.
The latent representation contains redundancy,
indicating the potential for bit savings by modeling such correlations.\par
However, existing entropy models fail to capture correlations in local spatial, global spatial, and channel domains.
Spatial context modules have limited interactions between channels, while channel-wise context modules
lack interaction within the current slice.
Moreover, extending these models to high-resolution image coding with \textit{acceptable} even \textit{linear} complexity is \textit{non-trivial}.
These challenges, along with the potential to enhance rate-distortion performance,
motivate us to design a \textit{linear} complexity  multi-reference entropy model.
Our proposed \textit{linear} complexity multi-reference entropy model effectively
captures correlations in local spatial, global spatial, and channel domains,
while maintaining a modest complexity for \textit{high-resolution} image coding.
Further details on our model are presented in the subsequent sections.
\begin{figure*}[t]
  \centering
  \includegraphics[width=\linewidth]
  {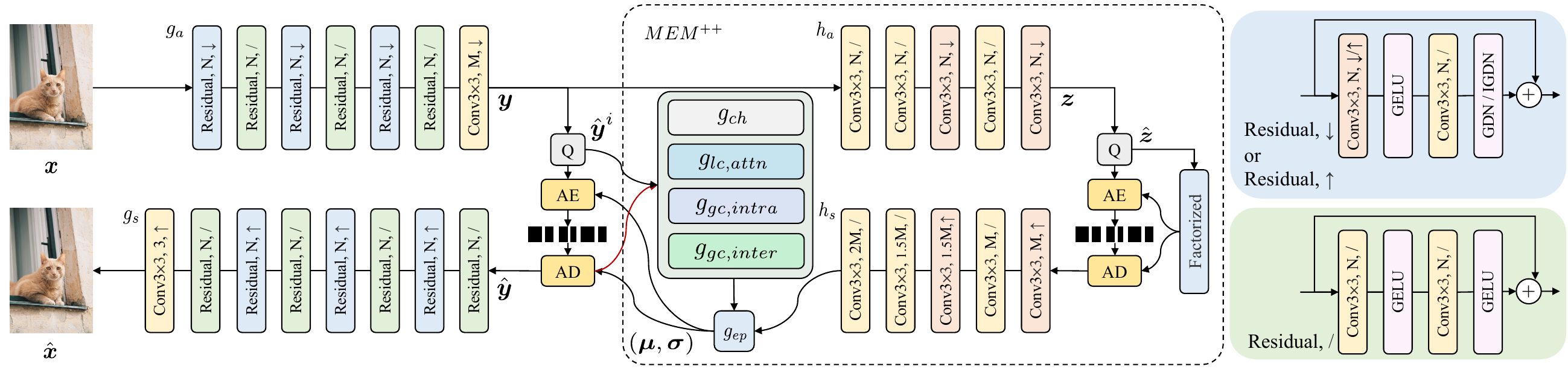}
  \caption{The overall architecture of MLIC$^{++}$.
  $\downarrow$ means down-sampling.
  $\uparrow$ means up-sampling.
  / means stride equals $1$.
  Red line is the dataflow during decoding.
  ${\boldsymbol{x}}$ is the input image and $\hat{\boldsymbol x}$ is the reconstructed image. $Q$ is quantization. $AE$ is arithmetic encoding. $AD$ is arithmetic decoding.
  $\boldsymbol{y}$ is the latent representation and $\hat{\boldsymbol{y}}$ is
  the quantized latent representation. $\hat{\boldsymbol y}^i$ is the $i$-th slice of $\hat{\boldsymbol{y}}$.}
  \label{fig:arch}
  \end{figure*}
  \subsection{Overall Architecture}
  \label{sec:method:overview}
  \subsubsection{MLIC$^{++}$}
The overall architecture of proposed model is illustrated in Fig.~\ref{fig:arch}.
This model is named MLIC$^{++}$ to distinguish it from MLIC, and MLIC$^+$, which are introduced in our conference version~\cite{jiang2022mlic}.
The architecture of MLIC$^{++}$, as depicted in Fig.~\ref{fig:arch}, incorporates the analysis transform $g_a$, synthesis transform $g_s$,
hyper analysis $h_a$, and hyper synthesis $h_s$, which are simplified versions of Cheng'20~\cite{cheng2020learned}.
To reduce complexity, attention modules are removed.
\begin{figure}[t]
  \centering
  \includegraphics[width=\linewidth]
  {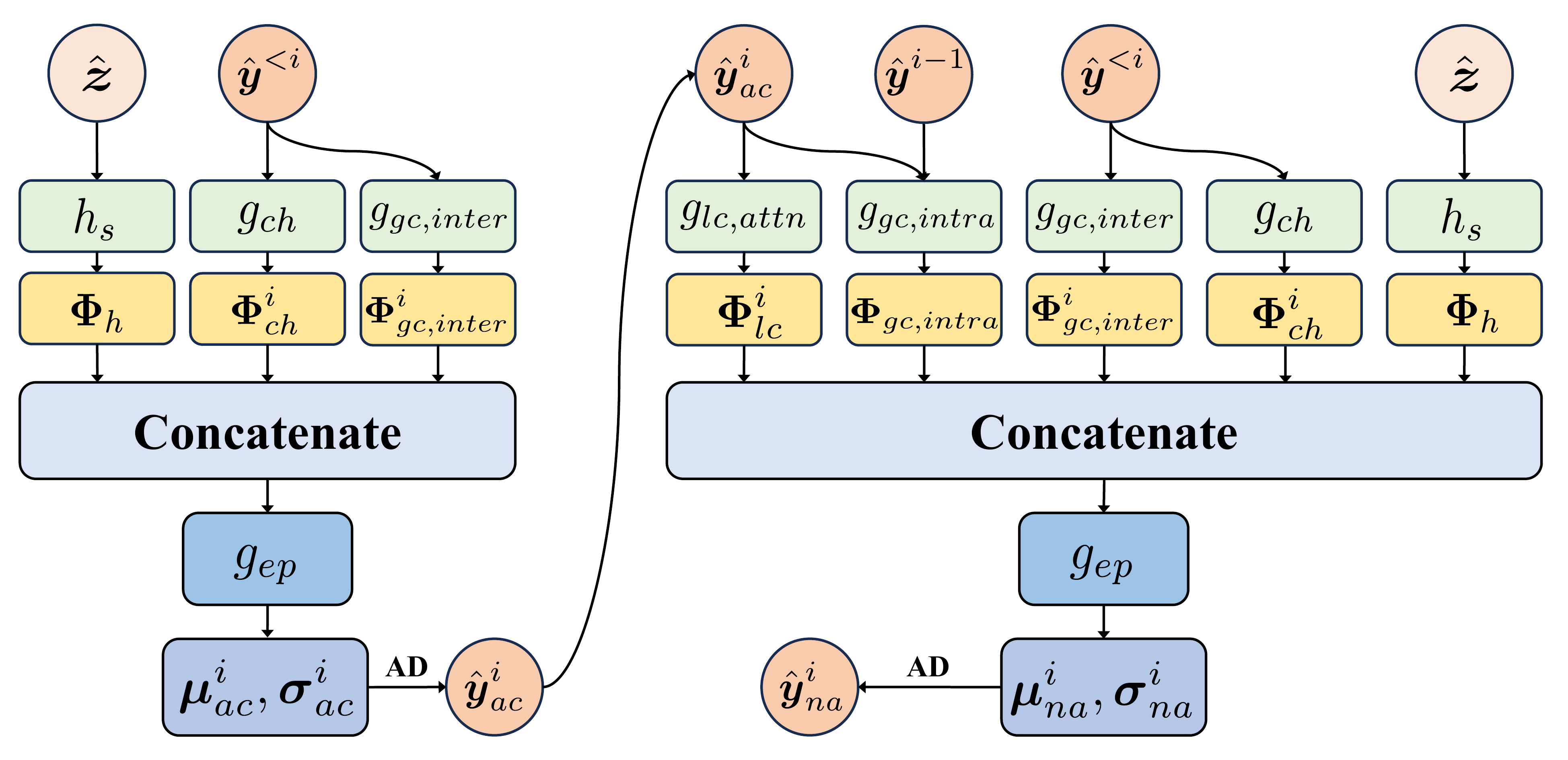}
  \caption{Linear Multi-Reference Entropy Model MEM$^{++}$.
  The figure illustrates the process of decoding a slice $\hat {\boldsymbol{y}}^i$.}
  \label{fig:mem}
\end{figure}
The hyper-parameters and settings of MLIC$^{++}$ are presented in Table~\ref{tab:settings}.
Same to Minnen \textit{et al.}~\cite{minnen2020channel},
we adopt \textit{mixed quantization}, which involves adding uniform noise for entropy estimation
and utilizing STE~\cite{theis2017lossy} to ensure differentiability in the quantization process.
Gaussian mean-scale distribution is adopted for entropy estimation.
For latent representation $\boldsymbol{y}$, the quantization and {estimated rate~\cite{theis2017lossy,yang2023introduction,balle2016end,balle2018variational,balle2020nonlinear}} is formulated as:
\begin{equation}
    \hat{\boldsymbol{y}} = \mathrm{STE}(\boldsymbol{y} - \boldsymbol{\mu}) + \boldsymbol{\mu},
  \end{equation} 
  \begin{equation} 
  \mathcal{R}_{\hat{\boldsymbol{y}}}=\mathbb{E}\left[ -\log \int_{-0.5}^{0.5} p({\boldsymbol{y}+\boldsymbol{u}})d\boldsymbol{u}  \right],
\end{equation}
where $\boldsymbol{\mu}$ is the estimated mean of latent representation $\boldsymbol{y}$, $\boldsymbol{u} \sim \mathcal{U}(-0.5,0.5)$, $\mathcal{R}_{\hat{\boldsymbol{y}}}$ is the estimated rate of
$\hat{\boldsymbol{y}}$.
\subsubsection{MEM$^{++}$}
The proposed linear complexity multi-reference entropy model effectively captures channel-wise, local spatial, and global spatial correlations with linear complexity.
The linear complexity entropy model is denoted as MEM$^{++}$ to distinguish it from MEM, and MEM$^{+}$,
which are proposed in our conference version~\cite{jiang2022mlic}.
To capture multi-correlations, the proposed MEM$^{++}$ consists of four components:
channel-wise context module $g_{ch}$, local spatial context module $g_{lc}$,
intra-slice global spatial context module $g_{gc,intra}$, and inter-slice global spatial context module $g_{gc,inter}$.
In the channel-wise context module, the latent representation $\hat{\boldsymbol{y}}$
is divided into slices $\{\hat {\boldsymbol{y}}^0, \hat {\boldsymbol{y}}^1, \cdots\, \hat {\boldsymbol{y}}^L\}$~\cite{minnen2020channel} along the channel dimension,
$L$ is the number of slices.
For the $i$-th slice $\hat {\boldsymbol{y}}^i$, the channel-wise context module captures
the channel-wise context $\boldsymbol{\Phi}_{ch}^i$ from slices $\hat {\boldsymbol{y}}^{<i}$.
To capture local spatial correlations, checkerboard pattern~\cite{he2021checkerboard} is employed, where
the latent representation $\hat{\boldsymbol{y}}^i$ is divided into
anchor part $\hat{\boldsymbol{y}}_{ac}^i$ and non-anchor part $\hat{\boldsymbol{y}}_{na}^i$.
$\hat{\boldsymbol{y}}_{ac}^i$ is local-context-free.
Local spatial context ${\boldsymbol{\Phi}}_{lc}^i$ of
$\hat{\boldsymbol{y}}_{na}^i$ is captured from $\hat{\boldsymbol{y}}_{ac}^i$.
We propose Overlapped Window-based Checkerboard Attention $g_{lc, attn}$ for better non-linearity and adaptability
to capture local spatial contexts.
The global contexts $\boldsymbol{\Phi}_{gc}$ of $i$-th slice are extracted from two dimensions:
intra-slice contexts $\boldsymbol{\Phi}_{gc, intra}^i$, and inter-slice contexts $\boldsymbol{\Phi}_{gc, inter}^i$.
We propose Intra-Slice Global Context Module $g_{gc, intra}^i$ and
Inter-Slice Global Context Module $g_{gc, inter}$ to capture such correlations.
Since different slices share the similar global similarity~\cite{jiang2022mlic,guo2021causal},
the global similarity of $\hat{\boldsymbol{y}}^{i-1}$ is employed
to predict the global correlations between $\hat{\boldsymbol{y}}^{i}_{ac}$ and $\hat{\boldsymbol{y}}^{i}_{na}$.
The inter-slice global context $\Phi^{i}_{gc,inter}$ are extracted
from slices $\hat{\boldsymbol{y}}^{< i}$ via the
global similarity of slices $\hat{\boldsymbol{y}}^{< i}$.
We introduce these modules in the following sections.
The structure of MEM$^{++}$ is illustrated in Table~\ref{tab:settings}.
We use Equation~\ref{eq:rd} as our loss function
and the {estimated rate~\cite{theis2017lossy,yang2023introduction,balle2018variational,balle2016end,balle2020nonlinear}} can be formulated as:
$\mathcal{R} =  \mathcal{R}_{\hat{\boldsymbol{z}}} + \sum^L_{i=0}\left(\mathcal{R}_{\hat{\boldsymbol{y}}^i_{ac}} + \mathcal{R}_{\hat{\boldsymbol{y}}^i_{na}}\right)$, where
\begin{equation}
  \mathcal{R}_{\hat{\boldsymbol{z}}}=\mathbb{E}\left[ -\log \int_{-0.5}^{0.5} p({\boldsymbol{z}+\boldsymbol{u}})d\boldsymbol{u}  \right],
\end{equation}
\begin{equation}
  \begin{aligned}
    \mathcal{R}_{\hat{\boldsymbol{y}}^i_{ac}} &= \mathbb{E}\biggl[-\log \int_{-0.5}^{0.5} p\biggl({\boldsymbol{y}}^i_{ac} + \boldsymbol{u}|\Phi_{h}, \\ & \Phi_{ch}^i, \Phi_{gc,inter}^i \biggr)d\boldsymbol{u} \biggr],
  \end{aligned}
  \end{equation}
  \begin{equation}
  \begin{aligned}
  \mathcal{R}_{\hat{\boldsymbol{y}}^i_{na}} &= \mathbb{E}\biggl[-\log \int_{-0.5}^{0.5} p\biggl({\boldsymbol{y}}^i_{na}+\boldsymbol{u}|\Phi_{h}, \\& \Phi_{lc}^i, \Phi_{ch}^i, \Phi_{gc,intra}^i, \Phi_{gc,inter}^i \biggr) d\boldsymbol{u}\biggr],
  \end{aligned}
\end{equation}
  $\mathcal{R}_{\hat{\boldsymbol{z}}}$ is the rate of side information,
  $\mathcal{R}_{\hat{\boldsymbol{y}}^i_{ac}}$ denotes the rate of the anchor
  part of $i$-th slice,   $\mathcal{R}_{\hat{\boldsymbol{y}}^i_{na}}$ denotes the rate of the non-anchor
  part of $i$-th slice,
$\boldsymbol{\Phi}_{h}$ is the hyper-priors extracted by hyper analysis $h_a$ and hyper synthesis $h_s$.\par
\subsection{Channel-wise Context Module}\label{sec:method:channel}
To extract channel-wise contexts, the latent representation $\hat {\boldsymbol{y}}$ is first evenly divided
into multiple slices
$\{\hat {\boldsymbol{y}}^0, \hat {\boldsymbol{y}}^1, \cdots, \hat {\boldsymbol{y}}^L\}$ along the channel dimension.
Slice $\hat {\boldsymbol{y}}^i$ is conditioned on
slices $\hat {\boldsymbol{y}}^{<i}$.
A channel context module $g_{ch}$ is employed to squeeze and
extract context information from $\hat {\boldsymbol{y}}^{<i}$ when
encoding and decoding $\hat {\boldsymbol{y}}^i$. $g_{ch}$ consists of
three $3\times 3$ convolutional layers. The channel context becomes
${\boldsymbol{{\boldsymbol{\Phi}}}}_{ch}^i = g_{ch}(\hat {\boldsymbol{y}}^{<i})$.
The channel-wise context module $g_{ch}$ is able to
refer to symbols in the same and close position in the previous
slices and helps select the most relative
channels and extract information beneficial for accurate probability estimation.
The channel number of each slice $S$ is a hyper-parameter.
Following Minnen \textit{et al}~\cite{minnen2020channel},
we set $S$ to $32$ and $L$ to $10$ in our model.
Following existing methods~\cite{minnen2020channel,zou2022the},
latent residual prediction (LRP) modules~\cite{minnen2020channel} are adopted
to predict quantization error according to
decoded slices and hyper-priors $\boldsymbol{\Phi}_{h}$.
Since the channel number of latent representation is freezed during training and
inference and the number of slices is quite small, the encoding speed
and decoding speed is still fast enough in spite of serial process among slices.
\begin{figure}[t]
  \centering
  \includegraphics[width=\linewidth]
  {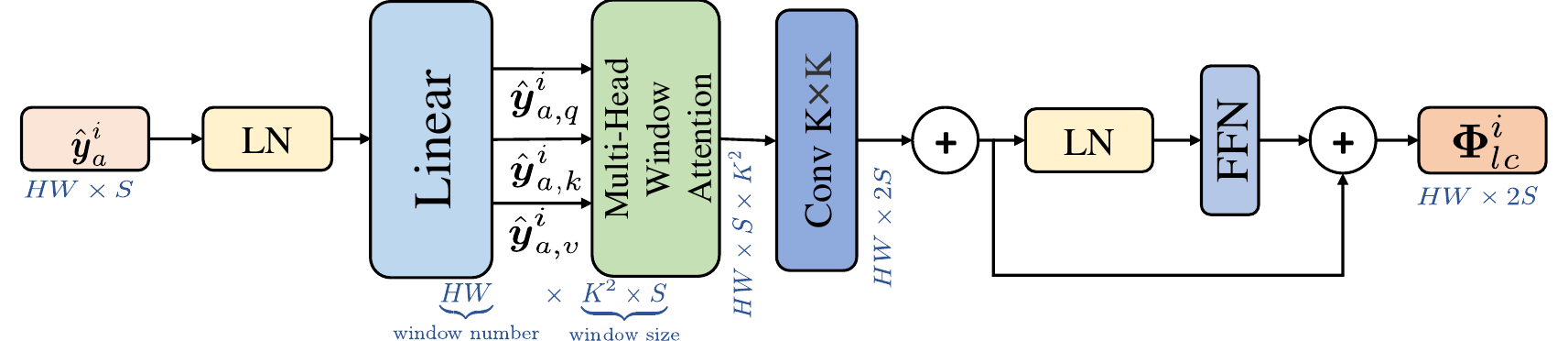}
  \caption{{Checkerboard Attention Context Module $g_{lc, attn}$.}}
  \label{fig:ckbd_attn_arch}
\end{figure}
\begin{figure}[t]
  \centering
  \includegraphics[width=\linewidth]
  {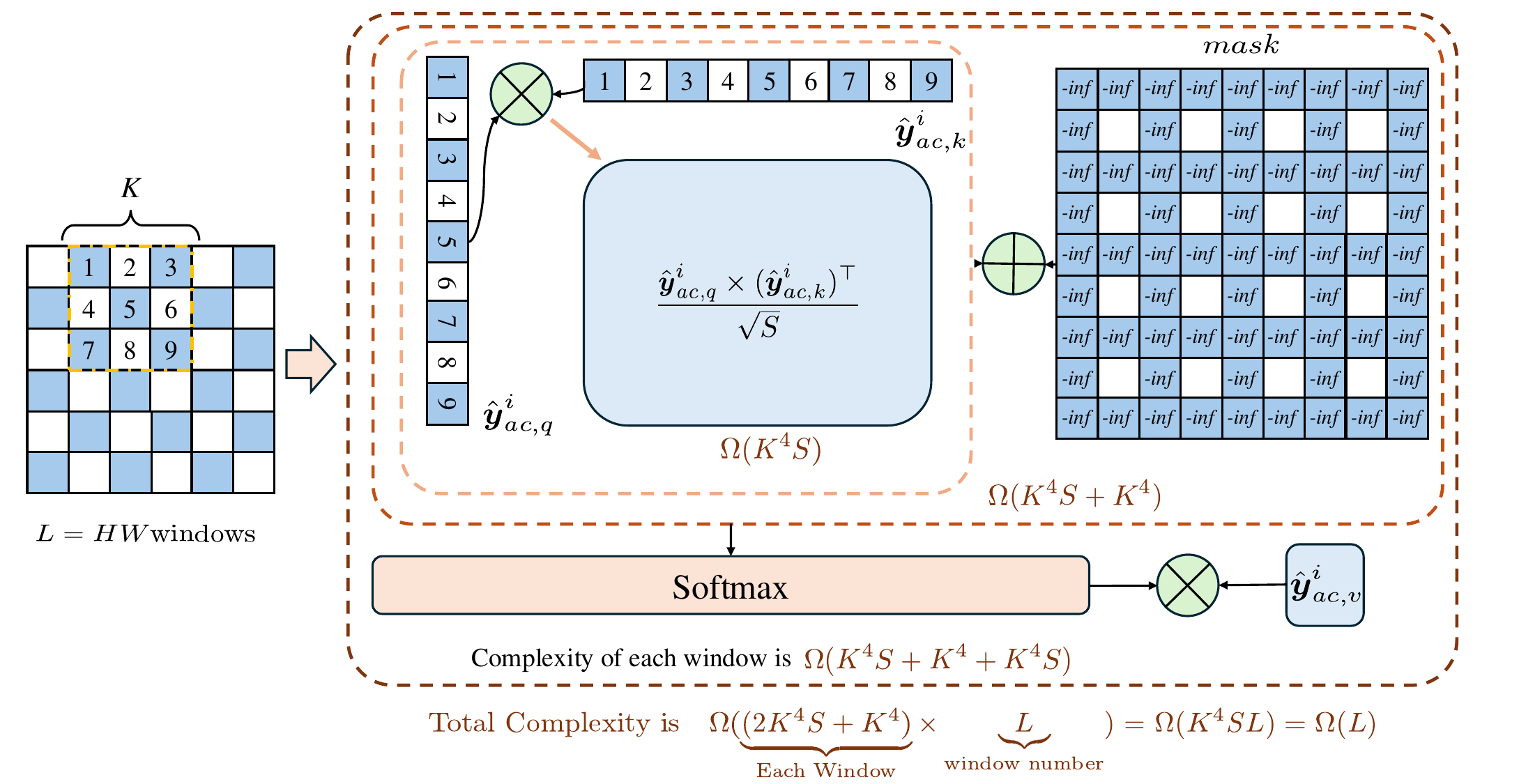}
  \caption{{Visualization of the process and complexity of Shifted Window-based Checkerboard Attention $g_{lc, attn}$.
  Blue squares are non-anchor part $\hat {\boldsymbol{y}}_{na}$,
  white squares are anchor part $\hat {\boldsymbol{y}}_{ac}$.}}
  \label{ckbd_attn}
  \end{figure}
\begin{figure}[t]
  \centering
  \includegraphics[width=0.8\linewidth]
  {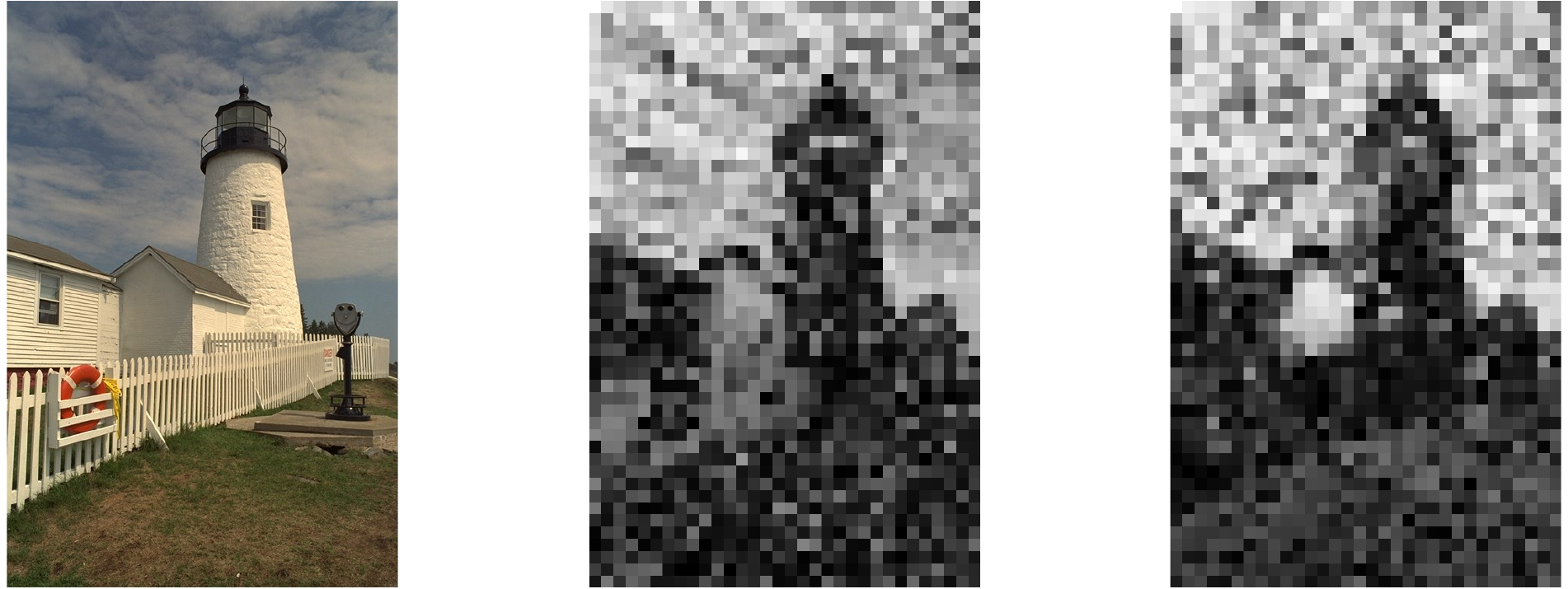}
  \caption{Cosine similarity in the spatial domain of different slices of latent representation of Kodim19 extracted by Cheng'20~\cite{cheng2020learned} (optimized for MSE, $\lambda=0.0483$)
  to visualize slices share similar global correlations. {The similarity maps are nearest-neighbor upsampled for visualization.}}
  \label{intra_cosine}
\end{figure}
\begin{figure*}
  \centering
  \includegraphics[width=\linewidth]
  {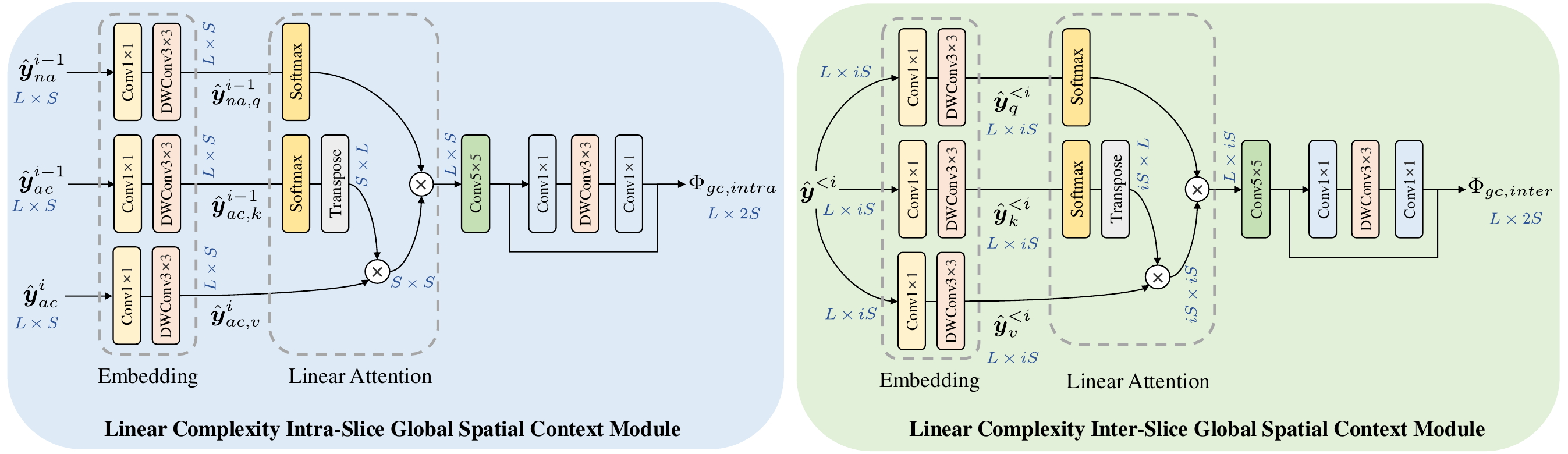}
  \caption{{Architectures of Linear Complexity Intra-Slice Global Spatial Context Module and
  Linear Complexity Inter-Slice Global Spatial Context Module.}}
  \label{fig:ggc}
  \end{figure*}
\subsection{Checkerboard Attention-based Local Context Module}
One limitation of CNN-based local context modules is their fixed weights,
which restricts their ability to capture content-adaptive contexts.
We argue that context-adaptation is essential due to the vast diversity of images.
In transformers~\cite{vas2017attention,dosovitskiy2020image,liu2021swin},
the attention weight is generated dynamically according to the input,
which inspires us to design a transformer-based content-adaptive local context module.
The local receptive field can be envisioned as a window,
where local spatial contexts are captured by dividing the feature map into windows.
Since each symbol is most relevant to the symbols around it, we propose to
make the divided windows \textit{overlapped}. To achieve this,
we propose the novel checkerboard attention context module $g_{lc, attn}$.
The process of $i$-th slice is taken as an example.
Assuming the resolution of the latent representation $\hat {\boldsymbol{y}}^i$ is $H \times W$,
the stride is set to $1$ to divide $\hat {\boldsymbol{y}}^i$
into $H \times W$  overlapped windows and the window size is $K\times K$.
To extract local correlations, the attention map of each window is computed at first.
Same as the convolutional checkerboard context module, interactions between $\boldsymbol{y}_{ac}^i$ and $\boldsymbol{y}_{na}^i$
and interactions in $\boldsymbol{y}_{na}^i$ are not allowed. An example of the attention mask is illustrated
in Fig.~\ref{ckbd_attn}. Importantly, this attention mechanism does not alter the resolution of each window.
Subsequently, a $K\times K$ convolutional layer is utilized to fuse local context information
{and} match the size of the local context with that of $\boldsymbol{y}^i$ before feeding
it to a feed-forward network (FFN)~\cite{vas2017attention}. The overall process is similar to standard transformer~\cite{vas2017attention}.
The process is formulated as:
\begin{equation}
    {{\hat {\boldsymbol{y}}}^i}_{attn} = \textrm{softmax}\left(\frac{{{\hat {\boldsymbol{y}}}^i}_{ac,q} \times ({{\hat {\boldsymbol{y}}}^i}_{ac,k})^\top}{\sqrt{S}} + mask\right) \times {{\hat {\boldsymbol{y}}}^i}_{ac,v},
  \end{equation}
    \begin{equation}
    {{\hat {\boldsymbol{y}}}^i}_{conv} = \textrm{conv}_{K\times K}({{\hat {\boldsymbol{y}}}^i}_{attn}),
  \end{equation}
  \begin{equation}
    {\boldsymbol{\Phi}}^i_{lc} = \textrm{FFN}({{\hat {\boldsymbol{y}}}^i}_{conv}) + {{\hat {\boldsymbol{y}}}^i}_{conv},
\end{equation}
where $\hat {\boldsymbol{y}}^i_{ac,q}, \hat {\boldsymbol{y}}^i_{ac,k}, \hat {\boldsymbol{y}}^i_{ac,v} = \textrm{Embed}(\hat {\boldsymbol{y}}^i_{ac})$,
$\hat {\boldsymbol{y}}^i_{ac}$ is anchor part of $i$-th slice, $mask$ the attention mask,
$S$ is the channel number of each slice, FFN is the feed-forward neural network~\cite{vas2017attention}.\par
Note that our overlapped window-partition is with \textit{linear} complexity, since the complexity of each window
is $\Omega(K^4)$.
The complexity of $g_{lc, attn}$ is $\Omega(K^4SL)$,
where $L=HW$, $S$ is the channel number of a slice.
\subsection{Linear Complexity Intra-Slice Global Context Module}
\label{sec:method:intra}
During the decoding process, it is challenging to determine the global
correlations between the current symbol and other symbols
due to the inherent \textit{encoding-decoding consistency}.
This is because the current symbol is unknown during decoding.
One potential solution is to embed the global correlations into the bit-stream,
but this approach introduces additional bits, thereby increasing the overall bit-rate.
Furthermore, in order to obtain precise global similarity,
it is necessary to calculate the similarity between the current symbol and \textit{all} other symbols,
which consumes a significant number of bits and is impractical to employ in real-world scenarios.
Consequently, representing global similarity with a limited number of bits or \textit{without} the need for additional bits
becomes a \textit{non-trivial} task.
In latent representation $\hat{\boldsymbol{y}} \in \mathbb{R}^{L\times C}$, where $L=H\times W$, $C$ is the channel number,
each channel contains distinct information, but they can be considered as thumbnails, as depicted in Fig.~\ref{fig:cosine}.
Notably, the channels exhibit \textit{similar} global similarities.
This is evident from the visualization of cosine similarities between two slices
of Cheng'20~\cite{cheng2020learned}, as visualized in Fig.~\ref{intra_cosine},
where despite differences in magnitude, the global correlations are similar.
When decoding the current slice $\hat {\boldsymbol{y}}^i \in \mathbb{R}^{L\times S}$, decoded slice
$\hat {\boldsymbol{y}}^{i-1} \in \mathbb{R}^{L\times S}$ assists in estimating the global correlations
in slice $\hat {\boldsymbol{y}}^{i}$.
However, a challenge arises in determining how to estimate these global correlations.
While cosine similarity may be useful,
it is fixed and may not accurately capture the features.
In this regard, attention maps prove to be a suitable choice.
The embedding layer is learnable, which make it flexible
to adjust the method for global correlations estimation by modifying queries, keys, and values.\par
First, the \textit{vanilla} approach is introduced.
The process of $i-1$-th slice and the $i$-th slice are taken as an example.
When compressing or decompressing $\hat {\boldsymbol{y}}^i$,
the correlations between anchor part $\hat {\boldsymbol{y}}^{i-1}_{ac}\in \mathbb{R}^{L\times S}$
and non-anchor part $\hat {\boldsymbol{y}}^{i-1}_{na}\in \mathbb{R}^{L\times S}$ of slice $\hat {\boldsymbol{y}}^{i-1}$ are first computed.
Because the checkerboard local context module makes anchor visible when decoding non-anchor part,
we multiply the anchor part of current slice $\hat {\boldsymbol{y}}^{i}_{a}$ with
the attention map between $\hat {\boldsymbol{y}}^{i-1}_{ac}$ and $\hat {\boldsymbol{y}}^{i-1}_{na}$,
which is employed as the approximation of global similarities between $\hat {\boldsymbol{y}}^{i}_{ac}$ and $\hat {\boldsymbol{y}}^{i}_{na}$.
Due to the local correlations, adjacent symbols have similar global correlations.
A $K\times K$ convolutional layer is employed to refine the attention
map by aggregating global similarities of adjacent symbols.
The process of this Intra-Slice Global Context $g_{gc, inter}$ is parallel and is formulated as:
\begin{equation}\label{eq:vanilla_intra}
\begin{aligned}
    &\hat {\boldsymbol{y}}^{i}_{attn} = \underbrace{\textrm{softmax}\left(\frac{{{\hat {\boldsymbol{y}}}^{i-1}}_{na,q} \times \left({{\hat {\boldsymbol{y}}}^{i-1}}_{ac,k}\right)^\top}{\sqrt{S}}\right)}_{\textrm{non-negative}} \times \hat {\boldsymbol{y}}^{i}_{ac,v},\\
    &\hat {\boldsymbol{y}}^{i}_{conv} = \textrm{conv}_{K\times K}(\hat {\boldsymbol{y}}^{i}_{attn}),\\
    &{\boldsymbol{\Phi}}^i_{gc,intra} = \textrm{DepthRB}(\hat {\boldsymbol{y}}^{i}_{conv}),
\end{aligned}
  \end{equation}
where $\hat {\boldsymbol{y}}^{i-1}_{na,q}, \hat {\boldsymbol{y}}^{i-1}_{ac,k} = \textrm{Embedding}\left(\hat {\boldsymbol{y}}^{i-1}\right)$, $\hat {\boldsymbol{y}}^{i}_{ac,v} = \textrm{Embedding}\left(\hat {\boldsymbol{y}}^i_{ac}\right)$,
Embedding is the embedding layer. Embedding layer consists of a $1\times 1$ convolutional layer and
a $3\times 3$ depth-wise convolutional layer. The $3\times 3$ depth-wise convolutional layer is
employed for learnable position embedding.
This is because the self attention is permutation-invariant and lacks inductive bias.
Using a depth-wise convolutional for position embedding does have serval benefits.
First, a $3\times 3$ depth-wise convolution is quite light, which has negligible
influences on overall complexity. Second, a convolution-based position embedding is flexible
for any resolution, due to its translation equivariance. Third, the convolution is
able to embed position information because of the zero-padding and the boundary effects~\cite{kayhan2020translation,islam2019much} of images.
$\textrm{DepthRB}$ is the depth-wise residual bottleneck~\cite{jiang2023slic} and is employed to enhance the non-linearity.\par
One drawback of \textit{vanilla} approach is its \textit{quadratic} complexity.
In Equation~\ref{eq:vanilla_intra}, the softmax operation specifies the order of tensor calculation.
The complexity of ${{\hat {\boldsymbol{y}}}^{i-1}}_{na,q} \times \left({{\hat {\boldsymbol{y}}}^{i-1}}_{a,k}\right)^\top$
is $O(L^2)$. The quadratic complexity leads to huge GPU memory consumption, longer
encoding and decoding time as illustrated in Fig.~\ref{fig:complex}, which makes it hard to employ the vanilla approach
for high-resolution image coding. In Equation~\ref{eq:vanilla_intra}, if $\left(\hat {\boldsymbol{y}}^{i-1}_{ac,k}\right)^\top\times \hat {\boldsymbol{y}}^{i}_{ac,v}$
is computed first, the overall complexity becomes $O(L)$, which is linear with the resolution.
Equation~\ref{eq:vanilla_intra} works because $0 < \textrm{softmax}\left(\frac{{{\hat {\boldsymbol{y}}}^{i-1}}_{na,q} \times \left({{\hat {\boldsymbol{y}}}^{i-1}}_{ac,k}\right)^\top}{\sqrt{S}}\right) < 1$.
The non-negativity makes it can be treated as a learnable similarity metric. If $\textrm{softmax}\left(\frac{{{\hat {\boldsymbol{y}}}^{i-1}}_{na,q} \times \left({{\hat {\boldsymbol{y}}}^{i-1}}_{ac,k}\right)^\top}{\sqrt{S}}\right) \to 0$,
${\hat {\boldsymbol{y}}}^{i-1}_{na,q}$ and ${\hat {\boldsymbol{y}}}^{i-1}_{na,q}$ are near orthogonal.
If $\textrm{softmax}\left(\frac{{{\hat {\boldsymbol{y}}}^{i-1}}_{na,q} \times \left({{\hat {\boldsymbol{y}}}^{i-1}}_{ac,k}\right)^\top}{\sqrt{S}}\right) \to 1$,
${\hat {\boldsymbol{y}}}^{i-1}_{na,q}$ and ${\hat {\boldsymbol{y}}}^{i-1}_{na,q}$ are very similar.
To solve the quadratic complexity, it is necessary to introduce a new operator which avoids the necessity to compute ${{\hat {\boldsymbol{y}}}^{i-1}}_{na,q} \times \left({{\hat {\boldsymbol{y}}}^{i-1}}_{a,k}\right)^\top$
first in practice while retaining the non-negativity. Efficient attention operation~\cite{shen2021efficient}
is introduced for non-negativity and linear complexity, which employ the softmax operation on $\hat{\boldsymbol{y}}^{i-1}_{na}$ in
row and the softmax operation on $\hat{\boldmath{y}}^{i-1}_{ac}$ in
column.
\begin{equation}
    \textrm{Attention} = \underbrace{\textrm{softmax}_2\left(\hat{\boldsymbol{y}}^{i-1}_{na,q}\right)\textrm{softmax}_1\left(\hat{\boldsymbol{y}}^{i-1}_{ac,k}\right)^{\top}}_{\textrm{non-negative}}\hat{\boldsymbol{y}}^i_{ac,v}.
    \label{eq:efficient}
\end{equation}
The process is illustrated in Equation~\ref{eq:efficient}.
In Equation~\ref{eq:efficient}, $\textrm{softmax}_2\left(\hat{\boldsymbol{y}}^{i-1}_{na,q}\right)\textrm{softmax}_1\left(\hat{\boldsymbol{y}}^{i-1}_{ac,k}\right)^{\top}$
is employed as the learnable similarity metric, where
 $0<\textrm{softmax}_2\left(\hat{\boldsymbol{y}}^{i-1}_{na,q}\right)< 1$, $0<\textrm{softmax}_1\left(\hat{\boldsymbol{y}}^{i-1}_{ac,k}\right)^{\top}< 1$,
 which makes $0<\textrm{softmax}_2\left(\hat{\boldsymbol{y}}^{i-1}_{na,q}\right)\\\textrm{softmax}_1\left(\hat{\boldsymbol{y}}^{i-1}_{ac,k}\right)^{\top}< 1$.
 The non-negativity makes
 $\textrm{softmax}_2\left(\hat{\boldsymbol{y}}^{i-1}_{na,q}\right)\textrm{softmax}_1\left(\hat{\boldsymbol{y}}^{i-1}_{ac,k}\right)^{\top}$
 can be employed as a similarity metric.
 If $\textrm{softmax}_2\left(\hat{\boldsymbol{y}}^{i-1}_{na,q}\right)\textrm{softmax}_1\left(\hat{\boldsymbol{y}}^{i-1}_{ac,k}\right)^{\top} \to 0$,
 ${\hat {\boldsymbol{y}}}^{i-1}_{na,q}$ and ${\hat {\boldsymbol{y}}}^{i-1}_{na,q}$ are near orthogonal.
 If $\textrm{softmax}_2\left(\hat{\boldsymbol{y}}^{i-1}_{na,q}\right)\textrm{softmax}_1\left(\hat{\boldsymbol{y}}^{i-1}_{ac,k}\right)^{\top} \\\to 1$,
 ${\hat {\boldsymbol{y}}}^{i-1}_{na,q}$ and ${\hat {\boldsymbol{y}}}^{i-1}_{na,q}$ are very similar.
 In addition, the $\textrm{softmax}_2\left(\hat{\boldsymbol{y}}^{i-1}_{na,q}\right)\textrm{softmax}_1\left(\hat{\boldsymbol{y}}^{i-1}_{ac,k}\right)$ is normalized,
 which makes each element can be treated probability.
 \par
 \begin{theorem}\cite{jiang2024ecvc}
  Same as the standard vanilla attention, 
  each row of the implicit similarity matrix 
  ${softmax}_2(\hat{\boldsymbol{y}}_{na,q}^i){softmax}_1(\hat{\boldsymbol{y}}_{na,k}^i)^\top$ 
  sums up to 1 and represents a normalized attention distribution over all positions.
  \end{theorem}
  \begin{proof}
  It is evident that each row of the the similarity matrix of the standard attention sums up to $1$.\par
  Let ${softmax}_2(\hat{\boldsymbol{y}}_{na,q}^i)=Q\in \mathbb{R}^{L\times C}$, 
  ${softmax}_1(\hat{\boldsymbol{y}}_{na,k}^i)^\top=K\in \mathbb{R}^{C\times L}$, where
  \begin{equation}
      Q=\begin{bmatrix}
  q_{1,1} & q_{1,2}  & \cdots   & q_{1,C}   \\
  q_{2,1} & q_{2,2}  & \cdots   & q_{2,C}  \\
  \vdots & \vdots  & \ddots   & \vdots  \\
  q_{L,1} & q_{L,2}  & \cdots\  & q_{L,C}  \\
       \end{bmatrix},
      \end{equation}
      \begin{equation}
  K^\top=\begin{bmatrix}
  k_{1,1} & k_{1,2}  & \cdots   & k_{1,L}   \\
  k_{2,1} & k_{2,2}  & \cdots   & k_{2,L}  \\
  \vdots & \vdots  & \ddots   & \vdots  \\
  k_{C,1} & k_{C,2}  & \cdots\  & k_{C,L}  \\
       \end{bmatrix}.
  \end{equation}
  Due to the characteristics of softmax operation, we can obtain 
  $\sum_{i=1}^C q_{j, i}=1$, where $1\leq j\leq L$; $\sum_{i=1}^L k_{j,i}=1$, 
  where $1\leq j\leq C$.
  \begin{equation}
      QK^\top=\begin{bmatrix}
  \sum_{i=1}^C q_{1,i}k_{i,1}   & \cdots   & \sum_{i=1}^C q_{1,i}k_{i,L}   \\
  \sum_{i=1}^C q_{2,i}k_{i,1}     & \cdots & \sum_{i=1}^C q_{2,i}k_{i,L}  \\
  \vdots   & \ddots   & \vdots  \\
  \sum_{i=1}^C q_{L,i}k_{i,1}  & \cdots\  & \sum_{i=1}^C q_{L,i}k_{i,L}  \\
       \end{bmatrix}.
  \end{equation}
  The sum of $\ell$-th row is taken as an example.
  \begin{equation}
  \begin{aligned}
      \textrm{Sum}_{\ell} &= \sum_{i=1}^C q_{\ell,i}k_{i,1} + \sum_{i=1}^C q_{\ell,i}k_{i,2} + \cdots + \sum_{i=1}^C q_{\ell,i}k_{i,L}\\
      &= \left(q_{\ell, 1}k_{1,1} + q_{\ell, 2}k_{2,1} + \cdots + q_{\ell, C}k_{C,1}\right) \\&+ \left(q_{\ell, 1}k_{1,2} + q_{\ell, 2}k_{2, 2} + \cdots + q_{\ell, C}k_{C,2}\right) + \cdots \\& + \left(q_{\ell, 1}k_{1,L} + q_{\ell, 2}k_{2, L} + \cdots + q_{\ell, C}k_{C,L}\right)\\
      &= q_{\ell, 1}\underbrace{\sum_{i=1}^Lk_{1,i}}_{1} + q_{\ell, 2}\underbrace{\sum_{i=1}^Lk_{2,i}}_{1} + \cdots + q_{\ell, C}\underbrace{\sum_{i=1}^Lk_{C,i}}_{1} \\&= \sum_{i=1}^C q_{\ell, i} =1.
  \end{aligned}
  \end{equation}
  \end{proof}
 The metric is \textit{implicit}  because there is no need to compute $\textrm{softmax}_2\left(\hat{\boldsymbol{y}}^{i-1}_{na,q}\right)\textrm{softmax}_1\left(\hat{\boldsymbol{y}}^{i-1}_{ac,k}\right)^{\top}$.
  Since we use softmax operation on $\hat{\boldsymbol{y}}^{i-1}_{na,q}$ and
  $\hat{\boldsymbol{y}}^{i-1}_{ac,k}$ separately,
  $\textrm{softmax}_1\left(\hat{\boldsymbol{y}}^{i-1}_{ac,k}\right)^{\top}\hat{\boldsymbol{y}}^{i}_{ac,v}$ can be computed first during training and testing.
  The complexity of it is $O(L)$, which is linear with the resolution.
  The linear complexity makes it easier to employ the global spatial context module
  for high-resolution image coding.
  \subsection{Linear Complexity Inter-Slice Global Context Module}
\label{sec:method:inter}
Because of the global correlations between slices,
intra-slice global context module is extended to the inter-slice global context.
For a symbol at current slice, the symbol in previous slices is employed at the same position as
the approximation of the symbol at current slice,
since there are correlations among slices.
The correlations among slices or channels are illustrated in Fig.~\ref{fig:cosine} and Fig.~\ref{intra_cosine}.
This approximation makes
the anchor part and non-anchor part benefit from more contexts.
The process of $i$-th slice is taken as an example.
Same as linear complexity intra-slice global context module, attention mechanism is
employed to measure the similarity. To make the inter-slice global
context capturing more efficient, the softmax operation in \textit{vanilla} attention
is divided into two independent softmax operations as discussed in Section~\ref{sec:method:intra}.
The overall process is
\begin{equation}\label{eq:vanilla_inter}
\begin{aligned}
  &\hat {\boldsymbol{y}}^{i}_{attn} = \underbrace{\textrm{softmax}_2\left(\hat{\boldsymbol{y}}^{< i}_q\right)\textrm{softmax}_1\left(\hat{\boldsymbol{y}}^{< i}_k\right)^{\top}}_{\textrm{non-negative}}\hat{\boldsymbol{y}}^{< i}_v,\\
  &\hat {\boldsymbol{y}}^{i}_{conv} = \textrm{conv}_{K\times K}(\hat {\boldsymbol{y}}^{i}_{attn}),\\
  &{\boldsymbol{\Phi}}^i_{gc,inter} = \textrm{DepthRB}(\hat {\boldsymbol{y}}^{i}_{conv}),
\end{aligned}
\end{equation}
where $\hat {\boldsymbol{y}}^{< i}_q, \hat {\boldsymbol{y}}^{< i}_{k}, \hat {\boldsymbol{y}}^{< i}_{v} = \textrm{Embedding}(\hat {\boldsymbol{y}}^{< i})$,
Embedding is the embedding layer. Embedding layer consists of a $1\times 1$ convolutional layer and
a $3\times 3$ depth-wise convolutional layer. The $3\times 3$ depth-wise convolutional layer is
employed for learnable position embedding.
$\textrm{DepthRB}$ is the depth-wise residual bottleneck~\cite{jiang2023slic} and is employed to enhance the non-linearity.
Since the $\textrm{softmax}_1\left(\hat{\boldsymbol{y}}^{< i}_k\right)^{\top}\hat{\boldsymbol{y}}^{< i}_v$ is computed
during training and testing, the similarity metric is \textit{implicit} and the overall complexity is $O(L)$, which is linear with the resolution.
  \section{Experiments}
  \label{sec:exp}
  \subsection{Implementation Details}
  \label{sec:exp:setup}
  \subsubsection{Training Dataset}
  Our training datasets contains $10^5$ images\footnote[1]{\tt\url{https://github.com/JiangWeibeta/MLIC/blob/main/train_list.txt}}\footnote[2]{\tt\url{https://huggingface.co/datasets/Whiteboat/MLIC-Train}}
  with resolutions exceeding $512\times 512$.
  These images are selected from ImageNet~\cite{deng2009imagenet},
  COCO 2017~\cite{lin2014microsoft}, DIV2K~\cite{Agustsson2017NTIRE2C}, and Flickr2K~\cite{lim2017enhanced}.
  To address the existing compression artifacts in JPEG images,
  we follow the approach of Ball{\'e} \textit{et al}~\cite{balle2018variational}
  by further down-sampling the JPEG images using a randomized factor.
  This downsampling process ensures that the minimum height or width of the images falls within the range of 512 to 584 pixels,
  effectively reducing the compression artifacts.
  \subsubsection{Training Strategy}
  MLIC$^{++}$ is built on Pytorch~\cite{paszke2019pytorch} and
  CompressAI~\cite{begaint2020compressai}.
  Following the settings of CompressAI~\cite{begaint2020compressai},
  we set $\lambda \in \{18, 35, 67, 130, 250, 483\} \times 10^{-4}$ for MSE
  and set $\lambda \in \{2.4, 4.58, 8.73, 16.64, 31.73, 60.5\}$ for Multi-Scale Structural Similarity (MS-SSIM)~\cite{wang2003multiscale}.
  The batch size is set to $32$ and models are trained on a single Tesla A100 GPU.
  Each model is trained with an Adam optimizer with $\beta_1=0.9, \beta_2=0.999$.
  We train each model for 2M steps.
  \begin{figure*}[t]
    \centering
    \subfloat{
      \includegraphics[scale=0.47]{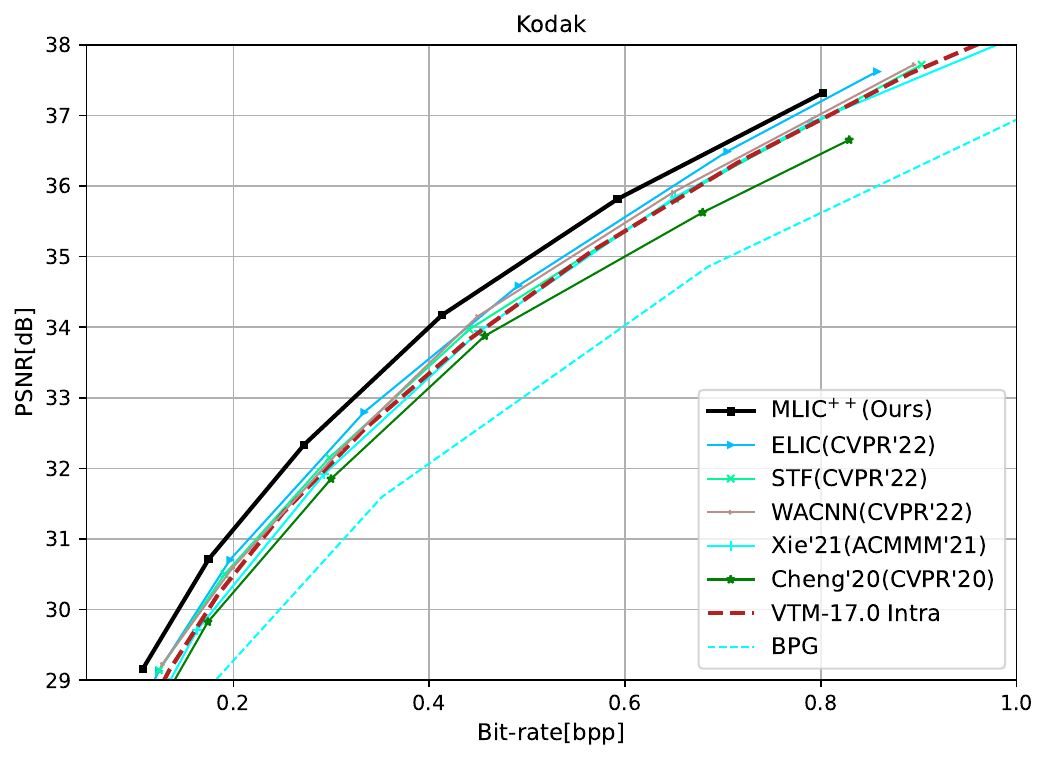}}
    \subfloat{
      \includegraphics[scale=0.47]{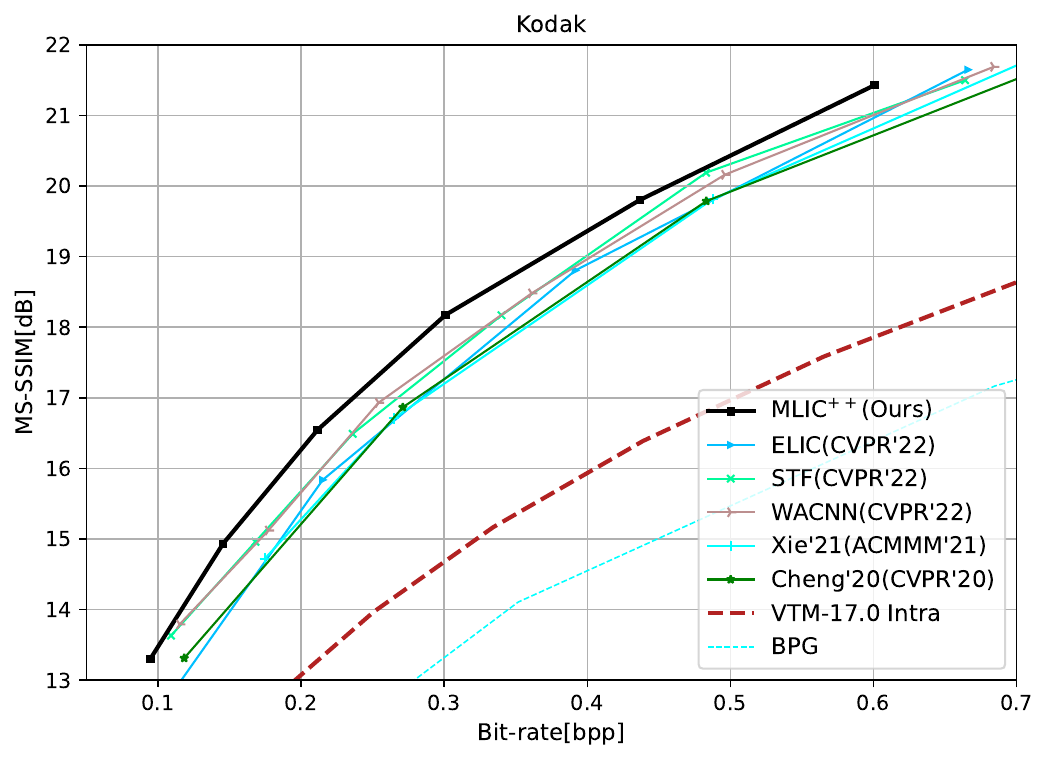}}\\
      \subfloat{
        \includegraphics[scale=0.47]{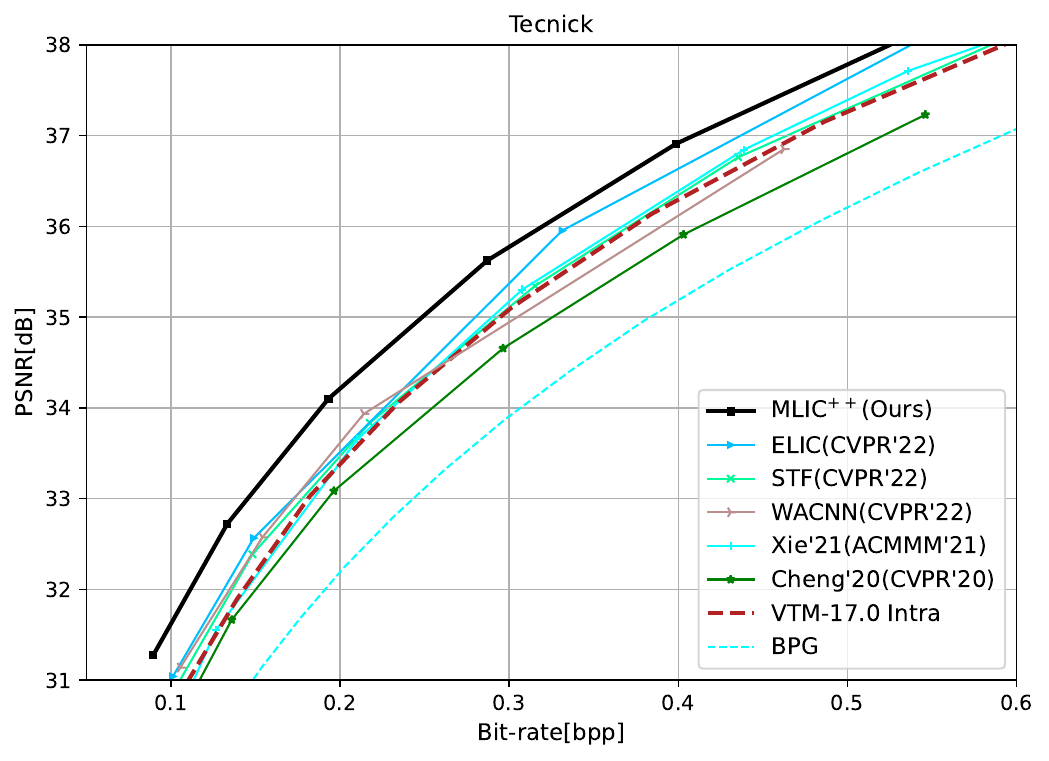}}
      \subfloat{
        \includegraphics[scale=0.47]{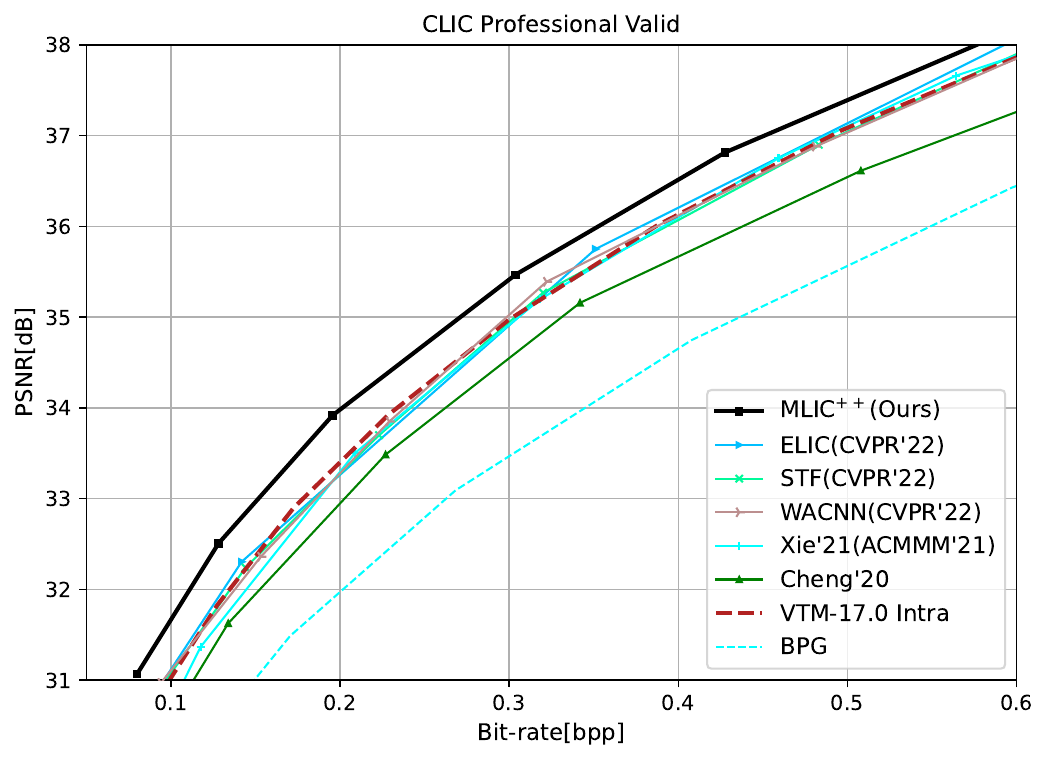}}
    \caption{PSNR-Bit-rate and MS-SSIM-Bit-rate curves.
    MS-SSIM is converted to dB for better clarity. Please zoom in for better view.}
    \label{fig:rd}
  \end{figure*}
The learning rate starts at $10^{-4}$ and drops to $3\times 10^{-5}$ at 1.5M steps,
drops to $10^{-5}$ at 1.8M steps, and drops to $3 \times 10^{-6}$ at 1.9M steps,
drops to $10^{-6}$ at 1.95M steps.
During training, we random crop images to $256\times 256$ patches during the
first $1.2$M steps. To further exploit the effectiveness of global context modules, images are cropped to $512\times 512$ patches
during the rest steps.
Large patches are beneficial for learning global references.
The latent representation can be sparse due to checkerboard partition.
A large latent representation makes it more difficult for the model to
capture the global contexts and improves the generalization ability of the model across different resolutions.
\subsection{Benchmarks and Metrics}
To thoroughly assess the generalization capability of the learned image compression models,
we conduct performance evaluations on three distinct datasets, including Kodak~\cite{kodak}, Tecnick~\cite{tecnick2014TESTIMAGES},
  CLIC Professional Valid~\cite{CLIC2020}.
  \begin{itemize}
  \item Kodak~\cite{kodak} is selected as a test set for almost
  all end-to-end image compression models~\cite{balle2016end,theis2017lossy, balle2018variational,minnen2018joint,cheng2020learned,minnen2020channel,chen2021nic,guo2021causal,wu2021learned,xie2021enhanced,gao2021neural,chen2022two,zou2022the,he2022elic,koyuncu2022contextformer,jiang2022mlic,duan2023lossy,liu2023learned}.
  It contains 24 raw $768\times 512$ images.
  \item Tecnick~\cite{tecnick2014TESTIMAGES} contains 100 $1200\times 1200$ RGB images,
  which is used as a test set in many methods~\cite{balle2018variational,minnen2018joint,xie2021enhanced,kim2022joint,jiang2022mlic,jiang2023slic,liu2023learned,zhu2022transformerbased,duan2023lossy,pan2022content}.
  \item CLIC Professional Valid~\cite{CLIC2020} is the validation set of 3rd Challenge on Learned Image Compression
  which contains $41$ images. Image in this dataset contains around $2048\times 1440$ pixels.
  CLIC Pro Valid is widely used in many recent methods~\cite{cheng2020learned,xie2021enhanced,he2022elic,zou2022the,liu2023learned,jiang2022mlic,wang2022neural,zhu2022unified,duan2023lossy}.
\end{itemize}
\par
We use the Bjøntegaard delta rate (BD-Rate)~\cite{bjontegaard2001calculation}
to evaluate the performance of learned image compression models.
\subsection{Rate-Distortion Performance}
\label{sec:exp:perf}
\subsubsection{Quantitive Results}
Rate-distortion curves are presented in Fig.~\ref{fig:rd}.
When compared with Cheng'20~\cite{cheng2020learned},
our proposed MLIC$^{++}$ can achieve a maximum
improvement of $0.5\sim0.8$ dB in PSNR and achieve a maximum improvement of $0.6$ dB in MS-SSIM in dB.
Our MLIC$^{++}$ adopts simplified analysis transform and synthesis transform
of Cheng'20~\cite{cheng2020learned}, therefore, the improvement of
model performance is attributed to our linear complexity multi-reference entropy
modeling. Our linear complexity multi-reference entropy models can capture
more contexts, which leads to much better rate-distortion performance.
The improvement also proves correlations exist in multiple dimensions
since Cheng'20~\cite{cheng2020learned} adopts an spatial auto-regressive context module.
Compared with ELIC~\cite{he2022elic}, our MLIC$^{++}$ can be up to $0.4$ dB higher at low bit rates in PSNR and
$1$ dB higher in MS-SSIM~\cite{wang2003multiscale}.\par
BD-rate reductions are presented in Table~\ref{tab:rd}.
\begin{table*}[t]
  \centering
  \footnotesize
  \setlength{\tabcolsep}{2.4mm}{
  \begin{tabular}{@{}cccccccccccccc@{}}
  \toprule
  \multicolumn{1}{c|}{\multirow{2}{*}{Methods}}                            & \multicolumn{2}{c}{Kodak}    & \multicolumn{2}{c}{Tecnick}  & \multicolumn{2}{c}{CLIC Professional Valid}  \\
  \multicolumn{1}{c|}{}                                                     & \multicolumn{1}{c}{PSNR} & \multicolumn{1}{c}{MS-SSIM}& \multicolumn{1}{c}{PSNR} & \multicolumn{1}{c}{MS-SSIM}& \multicolumn{1}{c}{PSNR} & \multicolumn{1}{c}{MS-SSIM} \\ \midrule
  \multicolumn{1}{c|}{VTM-17.0 Intra~\cite{bross2021vvc}}                                        & $0.00$       & $0.00$ & $0.00$       & $0.00$ & $0.00$       & $0.00$       \\\midrule
  \multicolumn{1}{c|}{Cheng'20 (CVPR'20)~\cite{cheng2020learned}}                                        & $+5.58$       & $-44.21$ & $+7.57$       & $-39.61$ & $+11.71$       & $-41.29$       \\\midrule
  \multicolumn{1}{c|}{Minnen'20 (ICIP'20)~\cite{minnen2020channel}}                                        & $+3.23$       & $-$ & $-0.88$       & $-$ & $-$       & $-$       \\\midrule
  \multicolumn{1}{c|}{Qian'21 (ICLR'21)~\cite{qian2020learning}}                                        & $+10.05$       & $-39.53$ & $+7.52$       & $-$ & $+0.02$       & $-$       \\\midrule
  \multicolumn{1}{c|}{Xie'21 (ACMMM'21)~\cite{xie2021enhanced}}                                        & $+1.55$       & $-43.39$ & $+3.21$       & $-$ & $+0.99$       & $-$       \\\midrule
  \multicolumn{1}{c|}{Guo'22 (TCSVT'22)~\cite{guo2021causal}}                                        & $-4.45$       & $-45.23$ & $-$       & $-$ & $-$       & $-$       \\\midrule
  \multicolumn{1}{c|}{LBHIC (TCSVT'22)~\cite{wu2021learned}}                                        & $-4.56$       & $-50.54$ & $-$       & $-$ & $-$       & $-$       \\\midrule
  \multicolumn{1}{c|}{Entroformer (ICLR'22)~\cite{qian2022entroformer}}                                        & $+4.73$       & $-42.64$ & $+2.31$       & $-$ & $-1.04$       & $-$       \\\midrule
  \multicolumn{1}{c|}{SwinT-Charm (ICLR'22)~\cite{zhu2022transformerbased}}                & $-1.73$      & $-$   & $-$& $-$& $-$& $-$  \\\midrule
  \multicolumn{1}{c|}{NeuralSyntax (CVPR'22)~\cite{wang2022neural}}                & $+8.97$      & $-39.56$   & $-$& $-$& $+5.64$& $-38.92$  \\\midrule
  \multicolumn{1}{c|}{McQuic (CVPR'22)~\cite{zhu2022unified}}                & $-1.57$      & $-47.94$   & $-$& $-$& $+6.82$& $-40.17$  \\\midrule
  \multicolumn{1}{c|}{STF (CVPR'22)~\cite{zou2022the}}         & $-2.48$      & $-47.72$  & $-2.75$      & $-$   & $+0.42$      & $-$      \\\midrule
  \multicolumn{1}{c|}{WACNN (CVPR'22)~\cite{zou2022the}}         & $-2.95$      & $-47.71$  & $-5.09$      & $-$   & $+0.04$      & $-$      \\\midrule
  \multicolumn{1}{c|}{ELIC (CVPR'22)~\cite{he2022elic}}                             & $-5.95$      & $-44.60$   & $-9.14$      & $-$ & $-3.45$      & $-$   \\\midrule
  \multicolumn{1}{c|}{Contextformer (ECCV'22)~\cite{koyuncu2022contextformer}}        & $-5.77$      & $-46.12$ & $-9.05$      & $-42.29$  & $-$      & $-$     \\\midrule
  \multicolumn{1}{c|}{Pan'22 (ECCV'22)~\cite{pan2022content}}        & $+7.56$      & $-36.20$ & $+3.97$      & $-$  & $-$      & $-$     \\\midrule
  \multicolumn{1}{c|}{NVTC (CVPR'23)~\cite{feng2023nvtc}}        & $-1.04$      & $-$ & $-$      & $-$  & $-3.61$      & $-$     \\\midrule
  \multicolumn{1}{c|}{LIC-TCM (CVPR'23)~\cite{liu2023learned}}        & $-10.14$      & $-48.94$ & $-11.47$      & $-$  & $-8.04$      & $-$     \\\midrule
  \multicolumn{1}{c|}{MLIC (ACMMM'23)~\cite{jiang2022mlic}}                                                    & $-8.05$      & $-49.13$ & $-12.73$      & $-47.26$ & $-8.79$      & $-45.79$    \\\midrule
  \multicolumn{1}{c|}{MLIC$^+$ (ACMMM'23)~\cite{jiang2022mlic}}                                                & $-11.39$      & $-52.75$    & $-16.38$      & $ -53.54$    & $-12.56$      & $-48.75$      \\\midrule
  \multicolumn{1}{c|}{{QARV (TPAMI'24)}~\cite{duan2023qarv}}                                                & {$+0.31$}      & {$-$}    & {$-3.03$}      & {$-$}    & {$-$}      & {$-$}      \\\midrule
  \multicolumn{1}{c|}{{FTIC (ICLR'24)}~\cite{li2024frequencyaware}}                                                & {$-12.99$}      & {$-51.13$}    & {$-14.88$}      & {$-$}    & {$-9.53$}      & {$-$}      \\\midrule
  \multicolumn{1}{c|}{{LLIC-TCM (TMM'24)}~\cite{jiang2023slic}}                                                & {$-10.94$}      & {$-49.73$}    & {$-14.99$}      & {$-$}    & {$-10.41$}      & {$-$}      \\\midrule
  \multicolumn{1}{c|}{MLIC$^{++}$ (Ours)}                                                & \bm{$-13.39$}      & \bm{$-53.63$}  & \bm{$-17.59$} &\bm{$-53.83$}& \bm{$-13.08$}& \bm{$-50.78$}   \\\midrule 
\end{tabular}}
\caption{BD-Rate $(\%)$ comparison for PSNR (dB) and MS-SSIM. The anchor is VTM-17.0 Intra. {``--'' means the result is not available}.} 
  \label{tab:rd}
\end{table*}
\begin{figure*}
  \centering
  \includegraphics[width=0.75\linewidth]
  {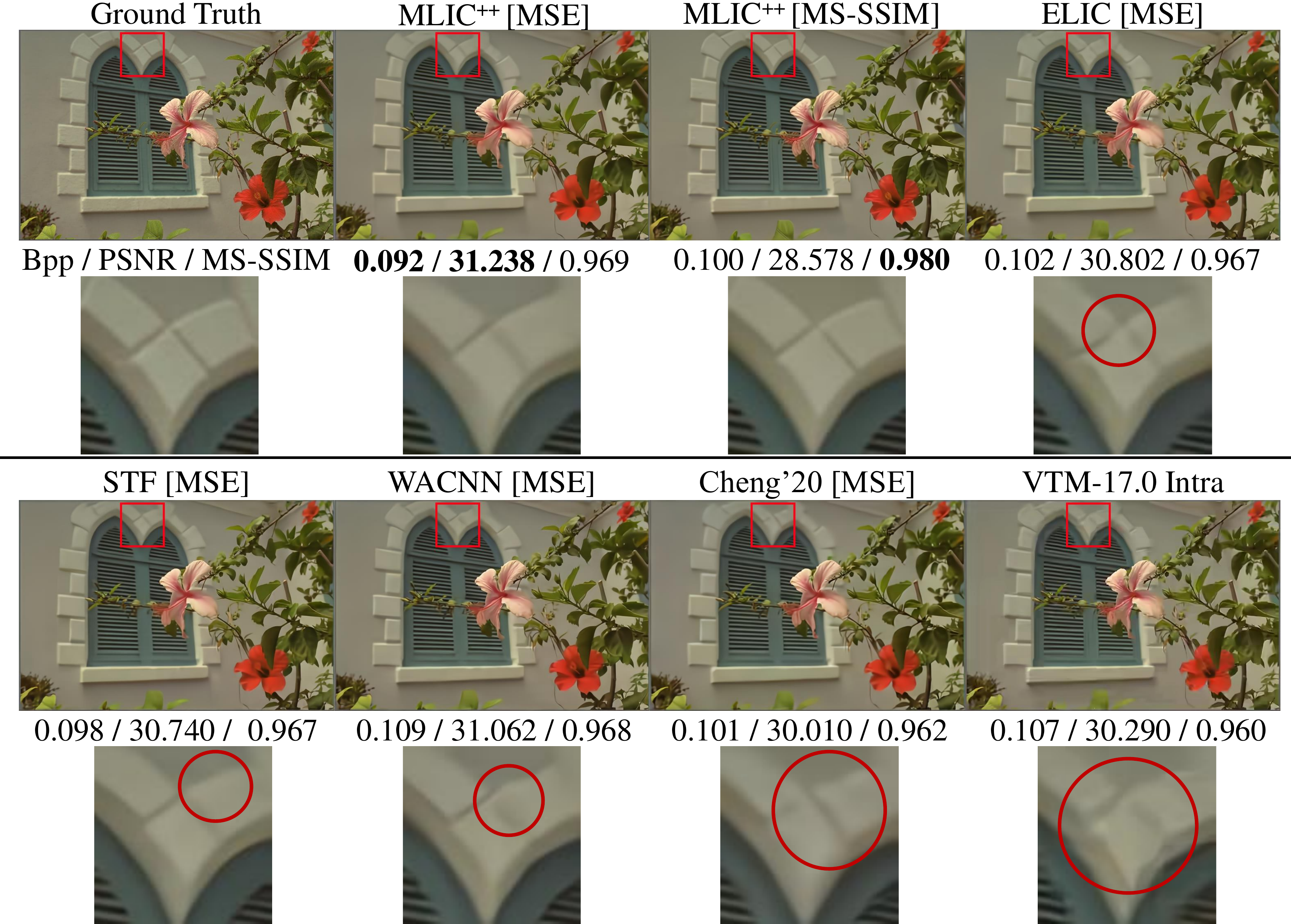}
  \caption{{Reconstructions of our proposed MLIC$^{++}$, recent learned image compression models~\cite{he2022elic,zou2022the,xie2021enhanced,cheng2020learned} and VTM-17.0 Intra~\cite{bross2021vvc}.
  “[MSE]” denotes the model is optimized for MSE and “[MS-SSIM]” denotes the model is optimized for MS-SSIM~\cite{wang2003multiscale}. Please zoom in for better view.}}
  \label{fig:vis}
  \end{figure*}
When computing BD-rate, VTM-17.0 Intra under YUV444 is employed as anchor. Unofficial weights of ELIC\footnote[3]{\tt\url{https://github.com/VincentChandelier/ELiC-ReImplemetation}} are used to
evaluate the rate-distortion performance of ELIC on Tecnick, CLIC Professional Valid.
Our MLIC$^{++}$ outperforms our previous MLIC and MLIC$^+$~\cite{jiang2022mlic}.
Our MLIC$^{++}$ reduces BD-rate by $13.39\%$ over VTM-17.0 Intra on Kodak while
existing method ELIC only reduces $5.95\%$ BD-rate and LIC-TCM
only reduces $10.14\%$ BD-rate on Kodak. Moreover, our MLIC$^{++}$ performs better
on high-resolution datasets, such as Tecnick
and CLIC Professional Valid. Our MLIC$^{++}$ reduces $17.59\%$ on Tecnick
and reduces $13.08\%$ BD-rate on CLIC Professional Valid, which are much better than existing methods.
Our MLIC$^{++}$ achieves state-of-the-art performance on all three datasets. The excellent performance of our MLIC$^{++}$ on these datasets
also demonstrates the excellent generalization of our MLIC$^{++}$.
We highlight the BD-rate for MS-SSIM, our proposed
MLIC$^{++}$ reduces about $50\%$ bits compared to VTM-17.0 Intra, which is a
large progress in learned image compression.
\subsubsection{Qualitive Results}
To further demonstrate superiority of our proposed MLIC$^{++}$, we compare our
MLIC$^{++}$ with learned image compression models Cheng'20~\cite{cheng2020learned},
Xie'21~\cite{xie2021enhanced}, STF~\cite{zou2022the}, WACNN~\cite{zou2022the},
ELIC~\cite{he2022elic} and non-neural codec VTM-17.0 Intra~\cite{bross2021vvc} on
perceptual quality. Fig.~\ref{fig:vis} presents the reconstructions of Kodim07 from Kodak.
PSNR value of the image reconstructed by our MLIC$^{++}$ optimized for MSE is $1$dB
  higher than image reconstructed by VTM-17.0 Intra~\cite{bross2021vvc}.
  MS-SSIM of the image reconstructed by our MLIC$^{++}$
  optimized for MS-SSIM is $0.02$ higher than image reconstructed by VTM-17.0 Intra.
  The windowsill of reconstructions are cropped to patches for clearer comparisons.
  {Upon closer inspection, certain regions within the images 
show more pronounced differences. For instance, in the {red-circled areas} of Figure~\ref{fig:vis}.
Moreover, our MSE-optimized MLIC$^{++}$ model achieves {lower bpp values—ranging from $84.4\%$ to $93.8\%$} of those in other models like ELIC, STF, 
WACNN, and Cheng'20—while maintaining competitive image quality. This demonstrates the efficiency and effectiveness of our approach.}
  \begin{figure*}[t]
    \centering
    \subfloat{
      \includegraphics[scale=0.41]{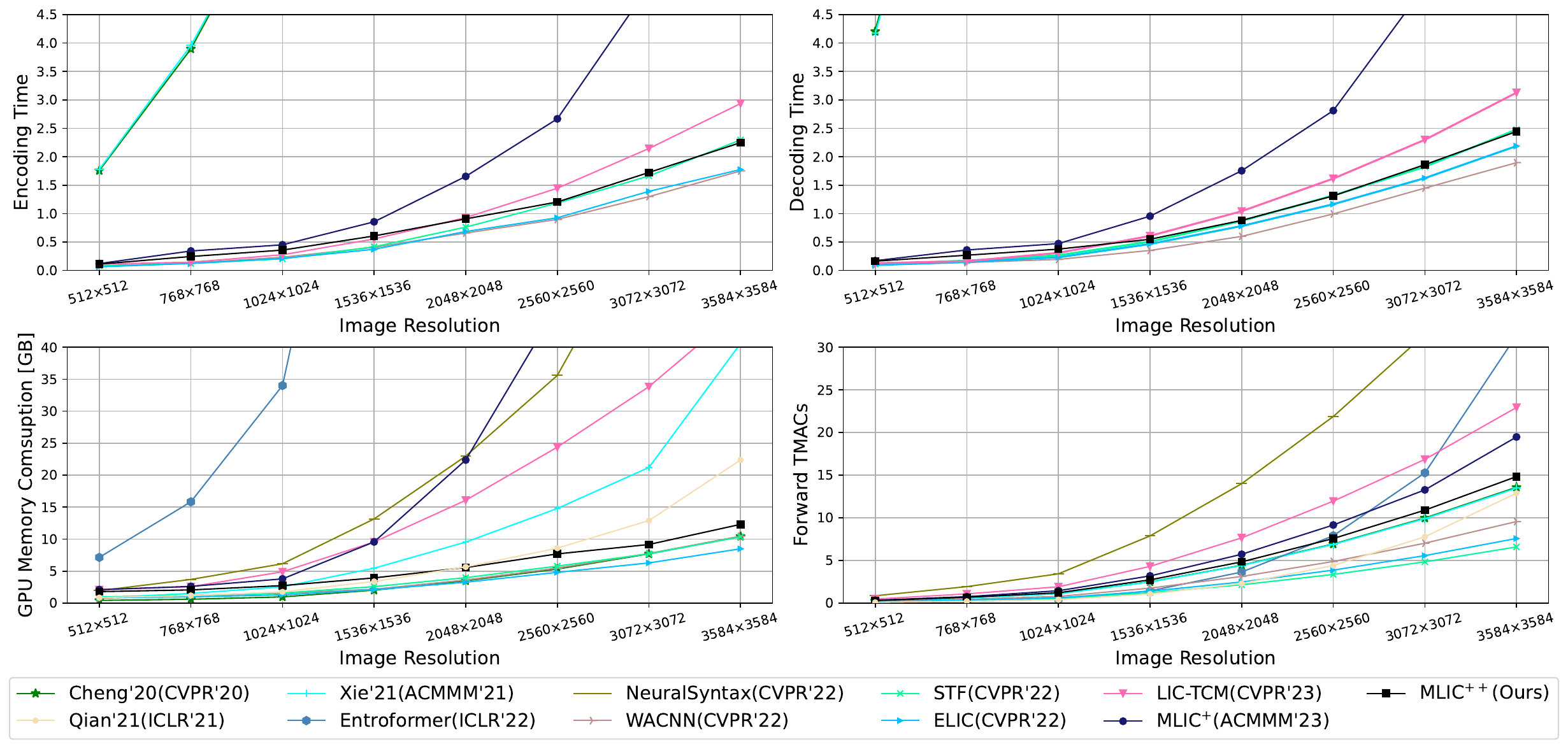}}
      \caption{{GPU memory consumption, encoding time, decoding time comparisons among our proposed models and recent learned image compression models~\cite{cheng2020learned,xie2021enhanced,qian2022entroformer,zou2022the,he2022elic,liu2023learned}.
      Encoding and decoding time of Qian'21~\cite{qian2020learning}, Entroformer~\cite{qian2022entroformer}, NeuralSyntax~\cite{wang2022neural} exceed 4.5s. Please zoom in for better view.}}
    \label{fig:complex}
  \end{figure*}
  \begin{table*}[!ht]
    \scriptsize
    \centering
    \setlength{\tabcolsep}{1.6mm}{
    \begin{tabular}{@{}cccccccccccccc@{}}
    \toprule
    \multicolumn{1}{c|}{\multirow{1}{*}{Context Modules}}   & \multicolumn{1}{c}{MLIC$^{++}$} & \multicolumn{1}{c}{Case 1}  & \multicolumn{1}{c}{Case 2}  & \multicolumn{1}{c}{Case 3} & \multicolumn{1}{c}{Case 4} & \multicolumn{1}{c}{Case 5}    & \multicolumn{1}{c}{Case 6}    & \multicolumn{1}{c}{Case 7}    & \multicolumn{1}{c}{Case 8}    \\\midrule
    \multicolumn{1}{c|}{Channel context module $g_{ch}$}                           & \checkmark                  & \checkmark    & \checkmark   & \checkmark            & \checkmark                & \checkmark                & \checkmark                & \checkmark                &         \\\midrule
    \multicolumn{1}{c|}{Checkerboard context module $g_{lc,ckbd}$}                      &                     &        &    &         &                   &                   & \checkmark                &                   &       \\\midrule
    \multicolumn{1}{c|}{Checkerboard attention context module $g_{lc,attn}$}                      & \checkmark                  & \checkmark    & \checkmark      & \checkmark          & \checkmark                & \checkmark                &                   &                   &         \\\midrule
    \multicolumn{1}{c|}{Linear intra-slice global spatial context module $g_{gc,intra}$}                     & \checkmark                  & \checkmark     &    &  \checkmark           & \checkmark                &                   &                   &                   &         \\\midrule
    \multicolumn{1}{c|}{Linear inter-slice global spatial context module $g_{gc,inter}$}                     & \checkmark                  &      & \checkmark    &              &                   &                   &                   &                   &         \\\midrule
    \multicolumn{1}{c|}{Position embedding}                 & \checkmark                  & \checkmark       & \checkmark        & \checkmark                  &                   &                   &                   &  &         \\\midrule
    \multicolumn{1}{c|}{DepthRB}                 & \checkmark                  & \checkmark       & \checkmark        &    & \checkmark               &                   &                   &                   &         \\\midrule
    \multicolumn{1}{c}{}                                    & \bm{$0.00$}                            & \bm{$+3.15$}  & \bm{$+2.79$}   & \bm{$+4.37$}                     & \bm{$+4.10$}        & \bm{$+12.92$}                                    &    \bm{$+15.24$}                           &  \bm{$+17.45$}                    & \bm{$+22.34$}        \\\midrule
    \end{tabular}}
    \caption{Ablation Studies on Kodak. BD-rate (\%) is employed to evaluate their contributions. MLIC$^{++}$ is the anchor.}
    \label{tab:ablation}
  \end{table*}
  \subsection{Computational Complexity}
  The computational complexity of models is measured in four aspects,
  including test GPU memory consumption, encoding time, decoding time, and forward Multiply-Accumulate operations (MACs).
  These metrics provide a comprehensive evaluation of the complexity from
  various perspectives, with particular emphasis on the first three metrics
  due to their direct relevance to real-world scenarios.
  It is worth noting that a model with lower MACs may
  exhibit \textit{slower} encoding and decoding speeds if the context module is \textit{serial} and consume
  a larger amount of GPU memory. Therefore,
  while MACs serves as an important measure of complexity,
  it should be considered in conjunction with the other
  metrics to obtain a more complete understanding of the characteristics of the model.
  \begin{figure*}[t]
    \centering
    \includegraphics[width=\linewidth]
    {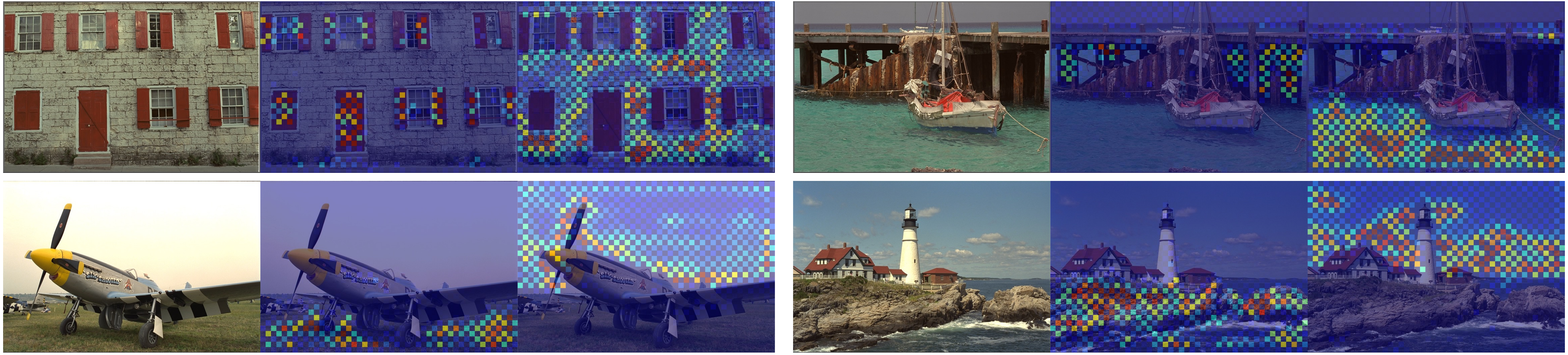}
    \caption{Intra-slice global spatial attention map $\textrm{softmax}_2\left(\hat{\boldsymbol{y}}^{i}_{na,q}\right)\left[j\right]\times \textrm{softmax}_1\left(\hat{\boldsymbol{y}}^{i}_{ac,k}\right)^\top$
    of Kodim01, Kodim11, Kodim20, Kodim21 from Kodak~\cite{kodak} dataset, where $i$ is the index of selected slice, $j$ is the index of selected query.
    Because interactions within anchor and non-anchor part are not allowed, the attention map is checkerboard-like.}
    \label{fig:intra_attn_map}
    \end{figure*}
    \begin{figure*}[t]
      \centering
      \includegraphics[width=\linewidth]
      {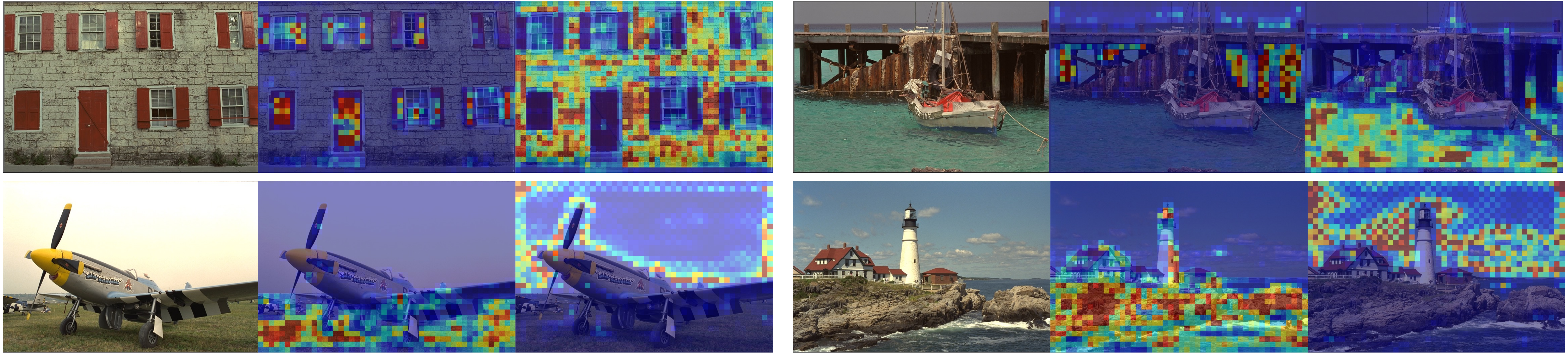}
      \caption{Inter-slice global spatial attention map $\textrm{softmax}_2\left(\hat{\boldsymbol{y}}^{i}_q\right)\left[j\right]\times \textrm{softmax}_1\left(\hat{\boldsymbol{y}}^{i}_k\right)^\top$
      of Kodim01, Kodim11, Kodim20, Kodim21 from Kodak~\cite{kodak} dataset, where $i$ is the index of selected slice, $j$ is the index of selected query.}
      \label{fig:inter_attn_map}
      \end{figure*}
    In order to better fit the real-world scenario,
    we test the model complexities in the case of different resolution images as inputs.
    We select $16$ images with resolution larger than $3584\times 3584$ from LIU4K test dataset~\cite{liu2020comprehensive} and we center crop
    these images to $\{512\times 512, 768\times 768, 1024\times 1024, 1536\times 1536,
    2048\times 2048, 2560\times 2560, 3072\times 3072, 3584\times 3584\}$ patches.
  We compared our MLIC$^{++}$ with our prior work MLIC$^{+}$ which employs \textit{vanilla} attention to illustrate the advantages of
    proposed \textit{linear} complexity global context capturing. We also compare our proposed MLIC$^{++}$ with
    recent learned image compression models~\cite{cheng2020learned,qian2020learning,xie2021enhanced,qian2022entroformer,wang2022neural,zou2022the,he2022elic,liu2023learned}.
    The experiments are conducted on a Tesla A100 GPU and a Xeon(R) Platinum 8260C CPU.
    The results are illustrated in Fig.~\ref{fig:complex}.\par
    \subsubsection{On GPU Memory comsuption}
    The quadratic complexity of vanilla attention leads to significantly more memory consumptions
    on high resolution image coding. When compressing $2048\times 2048$ images,
    MLIC$^{+}$ consumes nearly $22.37$ GB GPU memory. Our proposed linear complexity global
    context capturing significantly reduces the consumption of GPU memory as our proposed
    MLIC$^{++}$ only takes about $5.55$ GB GPU memory to compress a $2048\times 2048$ image.
    When compressing $3072\times 3072$ images, MLIC$^{+}$ consumes $46.38$ GB GPU memory while our
    proposed MLIC$^{++}$ only consumes $7.69$ GB GPU memory.
    Compared with recent LIC-TCM~\cite{liu2023learned}, our MLIC$^{++}$ consumes
    $\frac{1}{2}$ of the GPU memory consumed by LIC-TCM when compressing
    $2048\times 2048$ images.
    When compressing $2560\times 2560$ images, our MLIC$^{++}$ only consumes about $\frac{1}{4}$
    of the GPU memory consumed by LIC-TCM.
    When compressing a $3584\times 3584$ image, the GPU memory consumption of our proposed
    MLIC$^{++}$ is still only $12.3$ GB while the LIC-TCM requires $45.95$ GB GPU memory.
    The curve of the GPU memory consumed by our MLIC$^{++}$ as the resolution grows is also much flatter.
    \subsubsection{On Encoding and Decoding Time}
    When counting encoding time and decoding time, entropy coding and
    entropy decoding time are included. Since our MLIC$^{++}$ does not employ
    pixel-cnn-like spatial context capturing, our
    MLIC$^{++}$ encodes and decodes much faster than Cheng'20,
    Xie'21, and Entroformer. The linear complexity global context capturing leads to significant computational overhead reductions
    on high resolution images when compared with quadratic complexity based global context capturing.
    Compared with vanilla-attention based method,
    MLIC$^{++}$ encodes, decodes faster than MLIC$^{+}$ on $\{768\times 768, 1024\times 1024, 1536\times 1536, 2048\times 2048, 2560\times 2560\, 3072\times 3072, 3584\times 3584\}$ images.
    The time of MLIC$^{++}$ to encode $2560\times 2560$ images is about $\frac{1}{2}$ of
    the time of MLIC$^{+}$. Compared with recent LIC-TCM, MLIC$^{++}$
    needs more time to decode low resolution images while needs less time to encode high-resolution
    images, which can be attributed to fewer slices in LIC-TCM.
    At smaller resolutions, the bottleneck of encoding or decoding time
    is the number of slices, since the encoding and decoding of slices is serial, however,
    at larger resolutions, the bottleneck is no longer the number of slices
    but the overall computational complexity due to the high computational
    overhead of each slice.
    \subsubsection{On Forward MACs}
    Compared with MLIC$^{+}$,
    our MLIC$^{++}$, which employs the proposed linear complexity global spatial context
    modules, has lower MACs.
    Compared with the recent LIC-TCM~\cite{liu2023learned}, our MLIC$^{++}$ demonstrates significantly reduced MACs.
    Compared with ELIC~\cite{he2022elic}, STF~\cite{zou2022the}, and
    WACNN~\cite{zou2022the}, our MLIC$^{++}$ has higher MACs when the input is high-resolution image.
    One contributing factor to the higher MACs is the utilization of Cheng'20~\cite{cheng2020learned}
    as the basis for the transform module in our proposed MLIC$^{++}$.
    While the context module in Cheng'20 is simpler than that of ELIC,
    STF, and WACNN, the overall MACs of Cheng'20 are
    higher due to its transform modules having higher MACs.
    Since our MLIC$^{++}$ employs a modification of analysis transform and synthesis
    transform of Cheng'20, it results in higher MACs compared to ELIC and WACNN.
    In addition, since our MLIC$^{++}$ employed an more advanced entropy model,
    which leads to higher MACs. However, the entropy model has modest effect on the overall MACs
    since the input image is down-sampled for four times in analysis transform,
    which implies that designing more advanced entropy models is more resource-efficient.\par
    \subsection{Ablation Studies}
    \label{sec:ablation}
    \subsubsection{Settings}
    We conduct corresponding ablation studies and evaluate the contribution
    of each module on Kodak~\cite{kodak} dataset.
    Each model is optimized for MSE. We train each model for 2M steps. We use the training strategy in Section~\ref{sec:exp:setup}.
    \subsubsection{Analysis of Channel-wise Context Module}
    The inclusion of a channel-wise context module yields a substantial performance improvement
    when compared to Case 8, which solely incorporates hyper-priors, Case 7,
    incorporating channel-wise context modules, achieves a further reduction of $4.89\%$ in bit-rate.
    The channel-wise context module has the capability to reference symbols in the same and nearby positions
    in the preceding slices. The effectiveness of the channel-wise context
    module provides evidence of redundancy among channels.
    \subsubsection{Analysis of Local Spatial Context Module}
    The vanilla checkerboard context module leads to slight performance degradation~\cite{he2021checkerboard} compared
    with pixel-cnn-like serial context modules~\cite{van2016conditional,minnen2018joint}.
    The vanilla checkerboard context module contains one convolutional layer which is linear.
    The other drawback is the fixed kernel weights during inference.
    In contrast, our proposed checkerboard attention-based local spatial context module
    is non-linear and incorporates dynamic attention map generation with two-pass decoding,
    allowing for improved flexibility and adaptability. Our
    proposed checkerboard attention-based local spatial context module achieves a further reduction of
    $2.32\%$ in bit-rate compared to vanilla checkerboard context module. Compared with Case 7,
    which solely incorporates channel-wise context modules,
    models incorporating both local spatial and channel context modules demonstrate superior performance,
    further validating the presence of redundancy in the local spatial domain.
    \subsubsection{Analysis of Intra-Slice Global Context Module}
    The $i$-th slice is used as an example to illustrate the process.
    In practice, the computation of $\textrm{softmax}_1\left(\hat{\boldsymbol{y}}^{i-1}_{ac,k}\right)^{\top}\hat{\boldsymbol{y}}^i_{ac,v}$ in
    Equation~\ref{eq:efficient} is
    performed first for \textit{linear} complexity. However, we can still compute
    $\textrm{softmax}_2\left(\hat{\boldsymbol{y}}^{i}_{na,q}\right)\left[j\right]\times \textrm{softmax}_1\left(\hat{\boldsymbol{y}}^{i}_{ac,k}\right)^\top$
    as the attention map to validate the ability to capture global dependencies since
    the $\textrm{softmax}_2\left(\hat{\boldsymbol{y}}^{i}_{na,q}\right)\left[j\right]\times \textrm{softmax}_1\left(\hat{\boldsymbol{y}}^{i}_{ac,k}\right)^\top$
    is employed as the \textit{implicit} similarity metric, where $j$ is the index of selected query.
    The attention maps of Kodim01, Kodim11, Kodim 20, Kodim21 captured by proposed intra-slice global spatial context module $g_{gc,intra}$  is illustrated in Fig.~\ref{fig:intra_attn_map}.
    The checkerboard-like pattern in the attention map arises due to the absence of interactions within the anchor and non-anchor parts.
    Our model successfully captures distant correlations between the anchor and non-anchor parts,
    which local context modules are unable to achieve.
    Although our intra-slice global context module may bear some resemblance to cross-attention models,
    we focus solely on the interactions within a single slice.
    We only use the attention map of $\hat {\boldsymbol{y}}^{i-1}$ to predict correlations in $\hat {\boldsymbol{y}}^i$.
    When our proposed global context modules collaborate with local spatial context modules,
    the overall performance is further improved, underscoring the necessity of
    global spatial context modules for capturing global correlations
    and local spatial context modules for capturing local correlations.
    \subsubsection{Analysis of Inter-Slice Global Context Module}
    Taken $i$-th slice as an example,in linear complexity inter-slice global context module $g_{gc,inter}$
    computes $\textrm{softmax}_1\left(\hat{\boldsymbol{y}}^{i}_k\right)^\top\times \hat{\boldsymbol{y}}^{i}_v$ first.
    However, the $\textrm{softmax}_2\left(\hat{\boldsymbol{y}}^{i}_q\right)\left[j\right]\times \textrm{softmax}_1\left(\hat{\boldsymbol{y}}^{i}_k\right)^\top$
    can be computed as the attention map to validate the ability to capture inter-slice global dependencies, where $j$
    is the index of the selected query.
    The attention maps of Kodim01, Kodim11, Kodim20, Kodim21 captured by our inter-slice are illustrated in Fig.~\ref{fig:inter_attn_map}.
    The visualized attention maps clearly demonstrate the effective capture of global dependencies by our $g_{gc,inter}$,
    despite the model being trained in an \textit{implicit} manner.
    When $g_{gc, inter}$ collaborates with the intra-slice global context module $g_{gc,intra}$,
    teh local spatial context module, and the channel-wise context module, the rate-distortion
    performance of the model is further enhanced, which demonstrate the effectiveness
    of inter-slice global context module.
    \begin{figure}[t]
      \centering
      \includegraphics[width=\linewidth]
      {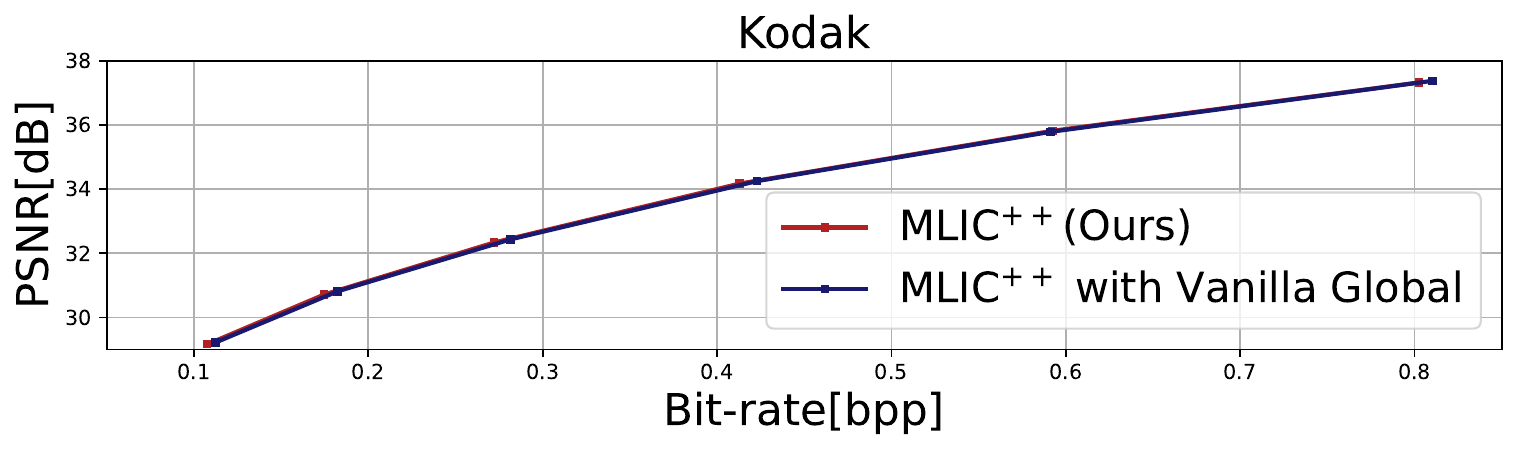}
      \caption{Rate-distortion performance comparisons between MLIC$^{++}$ and MLIC$^{{++}}$ with vanilla quadratic complexity global spatial context modules.}
      \label{fig:vanilla}
      \end{figure}
      \begin{table}[t]
        \centering
        \scriptsize
        \setlength{\tabcolsep}{4mm}{
        \begin{tabular}{@{}cccccccccccccc@{}}
        \toprule
        \multicolumn{1}{c|}{\multirow{1}{*}{Patch Size}}   & \multicolumn{1}{c}{$256\times 256$} & \multicolumn{1}{c}{$320\times 320$}   & \multicolumn{1}{c}{$512\times 512$}       \\\midrule
        \multicolumn{1}{c|}{}                                    & \bm{$0.00$}                            & \bm{$-3.66$}                               & \bm{$-6.12$}                                    \\\midrule
        \end{tabular}}
        \caption{BD-RATE(\%) of training using different patch size. The anchor is MLIC$^{++}$ trained on $256\times 256$ patches.}
        \label{tab:patch_size}
      \end{table}
      \begin{table}[t]
        \centering
        \scriptsize
        \setlength{\tabcolsep}{2mm}{
        \begin{tabular}{@{}cccccccccccccc@{}}
        \toprule
        \multicolumn{1}{c|}{\multirow{1}{*}{Context Module}}   & \multicolumn{1}{c}{$g_{ch}$} & \multicolumn{1}{c}{$g_{lc,attn}$}   & \multicolumn{1}{c}{$g_{gc,intra}$}   & \multicolumn{1}{c}{$g_{gc,inter}$}    \\\midrule
        \multicolumn{1}{c|}{KParams}                                    & \bm{$5810.11$}                            & \bm{$755.2$}                               & \bm{$732.5$}              &    \bm{$4633.34$}                     \\\midrule
        \multicolumn{1}{c|}{MMACs}                                    & \bm{$8925.18$}                            & \bm{$1148.2$}                               & \bm{$1116.2$}              &         \bm{$7100.50$}                \\\midrule
        \multicolumn{1}{c|}{{Inference Time (s)}}                                 & {\bm{$0.0036$}}                            & {\bm{$0.0094$}}                                  & {\bm{$0.0179$}}                  &         {\bm{$0.0130$}}                   \\\midrule  
      \end{tabular}}
      \caption{{Parameters, forward Macs and inference time of context modules on Kodak.}}
        \label{tab:context_complex}
      \end{table}
      \subsubsection{Analysis on Learnable Position Embedding and DepthRB}
      The position embedding and DepthRB lead to performance gains as illustrated in Table~\ref{tab:ablation}.
      Specifically, when DepthRB is not employed, a FFN is utilized instead. 
      The learnable position embedding is flexible because it is data-driven.
      Other position embedding method, Sinusoidal Position Embedding~\cite{vas2017attention},
      is tried and it leads negligible performance difference compared to model
      without position embedding. Relative Position Embedding~\cite{shaw2018self} is employed on
      attention map, which cannot be employed in our approach because
      the $\textrm{softmax}_1\left(\hat{\boldsymbol{y}}^{i-1}_{ac,k}\right)^{\top}\hat{\boldsymbol{y}}^{i}_{ac,v}$
      and $\textrm{softmax}_1\left(\hat{\boldsymbol{y}}^{< i}_k\right)^{\top}\hat{\boldsymbol{y}}^{< i}_v$ are
      computed first in our approach for linear complexity instead of the attention map
      $\textrm{softmax}_2\left(\hat{\boldsymbol{y}}^{i}_{na,q}\right)\left[j\right]\times \textrm{softmax}_1\left(\hat{\boldsymbol{y}}^{i}_{ac,k}\right)^\top$
      and $\textrm{softmax}_2\left(\hat{\boldsymbol{y}}^{i}_q\right)\left[j\right]\times \textrm{softmax}_1\left(\hat{\boldsymbol{y}}^{i}_k\right)^\top$.
      The learnable position embedding is not employed in our checkerboard attention-based local spatial context module $g_{lc,attn}$,
      because the performance improvement is quite negligible.
      In our $g_{lc,attn}$, the feature is partition into overlapped windows with zero padding.
      The zero padding and boundary effects imply that there is no need to insert a position embedding layer.
      \subsubsection{Comparisons with Model with Vanilla Global Spatial Context Modules}
      The comparison between MLIC$^{++}$ and MLIC$^{++}$ with vanilla quadratic complexity global spatial context modules is
      illustrated in Fig.~\ref{fig:vanilla}. It is evident that modeling global spatial contexts using linear complexity attention
      mechanisms does not result in performance degradation when compared to the vanilla attention mechanism.
      \subsubsection{Analysis of Training with Large Patches}
      In our training strategy, we use $512\times 512$ patches to train MLIC$^{++}$
      during the rest $0.8$M steps.
      We compare the differences in rate-distortion performance between different patch sizes
      $\{256\times 256, 320\times 320, 512\times 512\}$
      during the rest $0.8$M steps in Table~\ref{tab:patch_size}.
      Using large patches further boost the model performance.
      Using $256\times 256$ patches cannot fully exploit the performance of the model.
      The size of latent representation is $16\times 16$ if $256\times 256$ patches
      are adopted. $16\times 16$ latent representation is insufficient for model to learn
      long-range or global dependency as the resolutions of input images of the codec could be 2K or 4K.
      Moreover, checkerboard partition~\cite{he2021checkerboard} is employed, which
      makes the attention map sparse. Therefore, it is required to adopt large patches for better performance.
      Considering the overhead of training and model performance,
      adopting $512\times512$ patches is the best choice.
      \subsubsection{Comparisons among Different Context Modules}
      As illustrated in Table~\ref{tab:ablation}, the proposed global context modules
      $g_{gc,intra}$ and $g_{gc,inter}$ contribute most to performance,
      channel-wise context module $g_{ch}$ has the second highest contribution to performance, and
      local sptail context module $g_{lc,attn }$ has the lowest contribution to performance.
      Since each context module has a different role to perform, it is imperative that they work together
      for performance enhancement. The complexity of each context module is presented
      in Table~\ref{tab:context_complex}. The channel-wise context module
      has the most parameters and highest MACs.
      However, the MACs of the channel context module are only $1.77\%$ of the total MACs.
      The total MACs of all context modules are only $3.63\%$ of the total MACs.
      {Although $g_{ch}$ has high macs, its inference is fast due to the $3\times 3$ convolution being highly optimized on GPUs with CUDA/cuDNN\footnote[4]{\url{https://github.com/pytorch/pytorch/blob/main/aten/src/ATen/native/Convolution.cpp}}
while the window attention in $g_{lc,attn}$ and linear attention $g_{gc,intra}$ and $g_{gc,inter}$ are implemented in pure Python/Pytorch}\footnote[5]{\url{https://github.com/JiangWeibeta/MLIC/blob/main/MLIC\%2B\%2B/modules/transform/context.py}}.
{The $g_{gc,intra}$ consumes most time during inference could be attributed the partition of anchor part and non-anchor part.}
      \section{Conclusion}
      \label{sec:conclusion}
      In this paper, we propose a novel approach for capturing local spatial context using checkerboard attention,
      as well as linear complexity intra-slice and inter-slice global context modules,
      which significantly enhance the performance of the model while maintaining an acceptable \textit{linear} complexity.
      Based on proposed context modules, we propose linear complexity
      multi-reference entropy model MEM$^{++}$.
      Building upon MEM$^{++}$, we obtain state-of-the-art model MLIC$^{++}$.
      MLIC$^{++}$ exhibits \textit{linear} GPU memory consumption with resolution, 
      making it highly suitable for high-resolution image coding.
      To make our MLIC$^{++}$ more practical,
      our future work will focus on investigating the asymmetrical design~\cite{yang2023asymmetrically}
      between the analysis and synthesis transforms,
      as well as lighter linear complexity multi-reference entropy model.
\section*{Acknowledgement}
This work is financially supported by Guangdong Provincial Key Laboratory of Ultra High Definition Immersive Media Technology (Grant No. 2024B1212010006),
National Natural Science Foundation of China U21B2012, Shenzhen Science and Technology Program-Shenzhen Cultivation of Excellent Scientific and Technological Innovation Talents project (Grant No. RCJC20200714114435057), 
this work is also financially supported for Outstanding Talents Training Fund in Shenzhen.
\bibliography{mlicpp}
\bibliographystyle{icml2023}

\newpage
\appendix
\onecolumn


\end{document}